\pdfoutput=1
\documentclass[final,twoside,11pt]{entics}
\usepackage[final]{microtype}
\usepackage{amsmath}
\usepackage[capitalise]{cleveref}
\usepackage{graphicx}
\usepackage[all]{xy}
\usepackage{names}
\usepackage{diagram}
\usepackage{xparse}
\usepackage{expl3}
\usepackage{preamble}
\usepackage{thmtools}
\usepackage{thm-restate}
\usepackage{mathpartir}
\usepackage{enticsmacro}
\newtheorem{nota}[thm]{Notation}
\Crefname{lem}{Lemma}{Lemmas}
\Crefname{thm}{Theorem}{Theorems}
\Crefname{defn}{Definition}{Definitions}
\Crefname{cor}{Corollary}{Corollaries}
\Crefname{prop}{Proposition}{Propositions}
\crefformat{footnote}{#2\footnotemark[#1]#3}

\def\conf{MFPS 2024}
\def\myconf{mfps40} 	
\volume{4}			
\def\volu{4}			
\def\papno{18}			
\def\mydoi{14732}%
\def\lastname{Niu, Sterling, and Harper}
\def\runtitle{Cost-sensitive computational adequacy of higher-order recursion in synthetic domain theory}
\def\copynames{Y. Niu, J. Sterling, R. Harper}
\def\copyname{}

\begin{document}
\begin{frontmatter}
  \title{Cost-sensitive Computational Adequacy of Higher-order\\[1ex] Recursion in Synthetic Domain Theory}
  \author{Yue Niu\thanksref{a}\thanksref{yue}}
   \author{Jonathan Sterling\thanksref{b}\thanksref{jon}}
  \author{Robert Harper\thanksref{a}\thanksref{bob}}
   \address[a]{National Institute of Informatics\\  Tokyo, Japan}
   \thanks[yue]{Email: \href{mailto:yue\_niu@nii.ac.jp} {\texttt{\normalshape yue\_niu@nii.ac.jp}}}
  \address[b]{Department of Computer Science and Technology\\University of Cambridge\\ Cambridge, United Kingdom}
  \thanks[jon]{Email:  \href{mailto:js2878@cl.cam.ac.uk} {\texttt{\normalshape js2878@cl.cam.ac.uk}}}
   \address[c]{Computer Science Department\\ Carnegie Mellon University\\  Pittsburgh, USA}
  \thanks[bob]{Email:  \href{mailto:rwh@cs.cmu.edu} {\texttt{\normalshape rwh@cs.cmu.edu}}}
\begin{abstract}
  We study a cost-aware programming language for higher-order recursion dubbed \pcfc{} in the setting of \emph{synthetic domain theory} (SDT). Our main contribution relates the denotational cost semantics of \pcfc{} to its \emph{computational cost semantics}, a new kind of dynamic semantics for program execution that serves as a mathematically natural alternative to operational semantics in SDT. In particular we prove an internal, cost-sensitive version of Plotkin's computational adequacy theorem, giving a precise correspondence between the denotational and computational semantics for complete programs at base type. The constructions and proofs of this paper take place in the internal dependent type theory of an SDT topos extended by a \emph{phase distinction} in the sense of Sterling and Harper. By controlling the interpretation of cost structure via the phase distinction in the denotational semantics, we show that \pcfc{} programs also evince a noninterference property of cost and behavior. We verify the axioms of the type theory by means of a model construction based on relative sheaf models of SDT.
\end{abstract}
\end{frontmatter}

\section{Introduction}\label{intro}

In 1977 Plotkin~\cite{plotkin:1977} introduced a programming language for higher-order recursion, \textbf{PCF}, and defined \emph{computational adequacy}, a fundamental notion relating the denotational and operational semantics of a programming language in terms of the computational behavior of closed programs of ground type. Since Plotkin's seminal work, computational adequacy has been developed and refined to tailor various programming features and is an important property in the context of denotational approaches to program verification, enabling one to bring familiar mathematical structures and equational reasoning to bear on problems regarding program behavior. One refinement takes place in both the arrangement and metatheory of the semantics. In one direction, when the ambient category in which the model is defined has sufficient logical structure, one may view the object programming language from an internal perspective, a situation that was first thoroughly developed in the context of \emph{synthetic domain theory}~\cite{hyland:1991,phoa:1991}. In the other direction, one may move to a constructive metatheory such as the internal languages of generic elementary topoi or guarded type theory. The transition to these constructive, internal versions of computational adequacy has several benefits. First, by working inside topoi or type theories, one is automatically equipped with a powerful logical language to reason directly about domains and the denotational semantics, which was one of the original motivations for the development of synthetic domain theory. Second, the constructive nature of the metatheory means that one can argue that computationally adequate denotational semantics can be directly executed, as observed by De Jong~\cite{dejong:2023:thesis} in his work on constructive domain theory.

\emph{Computational adequacy in cost analysis}. More recently, Niu and Harper~\cite{niu-harper:2023} identified another application of internal adequacy in the context of cost analysis in \textbf{calf}~\cite{niu-sterling-grodin-harper:2022}, a type theory in which cost is reified as a computational effect and cost analysis takes the form of equational reasoning. A natural question in this setting is the relationship between the abstract cost bounds specified and proved in \textbf{calf} and the concrete bounds derived with respect to traditional operationally-based cost models. Viewing \calf{} as a semantic universe for \emph{internal denotational semantics}, Niu and Harper~\cite{niu-harper:2023} determined a criterion justifying the adequacy of an abstract \textbf{calf} cost model with respect to operational semantics by means of a cost-sensitive refinement of computational adequacy, which they instantiate using a variant of the Algol programming language, which featured higher-order (total) functions, first-order store and while loops.

\subsection{Motivations}

The aim of the present work is to extend the results of Niu and Harper~\cite{niu-harper:2023} to \emph{full higher-order recursion} by combining their ideas with \emph{synthetic domain theory}, culminating in a type theory for computing with \emph{generalized spaces} that support cost instrumentation alongside an information order for recursion. %

\subsubsection{Types as generalized spaces}\label{subsubsec:spaces}

The type theoretic philosophy is to study a real phenomenon in terms of \emph{abstract interfaces} rendered as types; type theoretic terms, then, correspond precisely to constructions in this discourse that preserve the geometrical and topological structure at play. For example, in synthetic domain theory (SDT), types have an intrinsic topology with respect to which all functions are continuous. By adhering to an abstract interface, the type-theoretic viewpoint brings to the fore the essential logical/geometric structure of underlying objects of interest and suppresses ill-defined objects and transformations, which makes it possible to define and reason about the domain of discourse in a conceptually simple manner that is completely rigorous.

Orthogonal to the topological and order-theoretic aspects of types as spaces in synthetic domain theory, we shall impose a different kind of geometric structure called a \emph{phase distinction}. Originating in the theory of program modules and logical relations~\cite{harper-mitchell-moggi:1990,sterling-harper:2021}, a phase distinction can generally refer to any situation in which types represent \emph{indexed} structures. For example, a program module can be represented as a family of dynamic runtime components whose ``shape'' is determined/indexed by a static module signature; similarly, a logical relation of a type theory can be represented as a family of computability data indexed in a syntactic type.
In this work, we are concerned with a similar kind of phase distinction called the \emph{intension-extension} phase distinction, introduced by Niu~\etal~\cite{niu-sterling-grodin-harper:2022} to enable simultaneous reasoning about both the cost and extensional behavior of programs.

The geometric aspect of phase distinctions emerges through two operations, \emph{restriction} and \emph{sealing}, that manipulate types to either trivialize the fibers or the base of the indexed structure represented by a given type. In the present case, these two operations are used to implement a form of cost profiling semantics for programs that can be easily ``stripped away'' via restriction, enabling one to simultaneously and faithfully carry out cost analysis in one framework; we shall describe this in mathematical terms in \cref{subsubsec:phase-distinction}.

In this paper we combine the intension-extension phase distinction of Niu~\etal~\cite{niu-sterling-grodin-harper:2022} with the domain-theoretic structure of SDT to obtain an intrinsically cost-sensitive theory of higher-order recursion in type theory. There are two main motivations for this combination. First, building on the work of Niu~\etal~\cite{niu-sterling-grodin-harper:2022,niu-harper:2023} we obtain an account of cost analysis in type theory that is compatible with general recursive programs, which we anticipate will enable more natural programming and verification techniques. Second, we will use the resulting type theory to define and study the \emph{internal denotational semantics} of \pcfc{}, a cost-aware version of \pcf{}, and prove a cost-sensitive internal computational adequacy theorem in the sense of \opcit. We explain what this entails in the following section.

\subsubsection{Internal denotational semantics}\label{subsubsec:internal}

By \emph{internal} denotational semantics, we mean to specify the syntax and semantics of a given programming language $\mathcal{L}$ called the \emph{object language} as constructions internal to type theory. Typically, one may define the syntax of the object language by means of a first-order encoding in terms of (indexed) inductive data types in which we have a type $\kw{tp} : \mathcal{U}$ of object-level types and a type family $\kw{tm} : \kw{tp} \to \mathcal{U}$ of object-level terms. A denotational model of $\mathcal{L}$ is given by a pair of maps $\sem{-}_0 : \kw{tp} \to \mathcal{U}$ and $\sem{-}_1 : (A : \kw{tp}) \to \kw{tm}(A) \to \sem{A}$ sending types to semantic domains and terms to elements of semantic domains. Usually we will use $\sem{-}$ to refer to either map when the context is unambiguous. Similarly, one may define an internal operational semantics of $\mathcal{L}$ as a type family ${\Downarrow} : (A : \kw{tp}) \to \kw{tm}(A) \to \kw{tm}(A) \to \mathcal{U}$ such that the type $e \Downarrow v$ (viewed as a proposition) holds just when $e$ evaluates to $v$.

Computational adequacy, first introduced in Plotkin~\cite{plotkin:1977}, is a statement that relates the operational and denotational semantics of $\mathcal{L}$ at the type of booleans (or any other observable type): a denotational semantics is \emph{computationally adequate} when $\sem{e} = b$ just when $e \Downarrow \overline{b}$, where $\overline{b}$ is the object-level boolean denoted a semantic boolean $b$. In the context of program verification, computational adequacy is a desirable property because the agreement of both semantics at observable types can be used to justify semantic reasoning with respect to operational behavior. In the context of cost analysis, the significance of this property is made explicit in the work of Niu and Harper~\cite{niu-harper:2023}, where the authors illustrate how a cost-sensitive refinement of internal computational adequacy can serve as the basis for validating user-defined cost models and program cost instrumentations. At a high level, results of this kind serve to ground the abstract reasoning supported by type theory to ``reality'' as given by operational semantics; we defer to \emph{op.~cit.} for a deeper discussion.

In this paper we will generalize the result of Niu and Harper~\cite{niu-harper:2023} to higher-order recursion. In contrast to \opcit{}, in this paper we prefer to work with a version of \pcfc{} whose execution model is given by a \emph{computational semantics} (see \cref{sec:comp-sem}). We will elaborate on the relationship between this new kind of dynamic semantics and traditional operational semantics in \cref{sec:discussion}.

\subsection{Mathematical techniques}

In constructing a model for the type theory supporting both a phase distinction and recursion, we employ mathematical tools that deserve some explication, which we summarise below.

\subsubsection{Cost as a phase distinction}\label{subsubsec:phase-distinction}

The writer monad on a monoid is a well-established way to express the computational effect of incurring and accumulating the cost of program execution. However, as pointed out by Niu~\etal~\cite{niu-sterling-grodin-harper:2022}, this arrangement does not faithfully model the semantics of cost profiling because it allows functions to branch on the cost component of their inputs. Thus, semantic functions lack a coherent \emph{underlying} behavior independent of the profiling, which is indispensable for stating behavioral or correctness specifications of algorithms.

In the \calf{} type theory of Niu~\etal~\cite{niu-sterling-grodin-harper:2022} this problem is resolved by means of the \emph{intension-extension phase distinction}, as discussed in \cref{subsubsec:spaces}. Mathematically, a type-theoretic \emph{phase} is simply a distinguished proposition \(\P\) whose associated \emph{open} and \emph{closed} modalities~\cite{rijke-shulman-spitters:2020} generate subuniverses for classifying purely extensional and purely intensional types, respectively. A \emph{purely extensional} type $A$ does not detect the presence of the phase proposition in the sense that the map $A \to (\P \to A)$ sending $a$ to the constant map $\P \to A$ determined by $a$ is an isomorphism; in other words a purely extensional type classifies only the behavior of programs. By contrast, a \emph{purely intensional} type $A$ ``collapses'' extensionally in the sense that $(\P \to A) \cong 1$. Given any type $A$, one may extract from $A$ the purely extensional and purely intensional part of $A$ by means of \emph{idempotent monadic} modalities. The \emph{extensional}/\emph{restriction} modality is simply defined as the function space out of the proposition: \(\P   \to  A\) (\ie the reader monad on $\P$). The \emph{intensional}/\emph{sealing} modality $\P \vee -$ is defined as the following pushout/quotient inductive type (QIT)~\cite{fiore-pitts-steenkamp:2021}: 
\begin{minipage}{0.5\textwidth}
  \DiagramSquare{
    nw = A \times \P,
    sw = A, 
    ne = \P, 
    se = \P \vee A,
    se/style = pushout,
    north = \pi_2,
    west = \pi_1,
  } 
\end{minipage}%
\begin{minipage}{0.5\textwidth}
  \iblock{
    \mhang{\textbf{inductive}~\P \vee A : \kw{Set}~\textbf{where}}{
      \mrow{\eta_{\P \vee -} : A \to \P \vee A}
      \mrow{\ast : \P \to \P \vee A}
      \mrow{\_ : (a : A) \to (u : \P) \to \eta_{\P \vee -}(a) = \ast(u)}
    }
  }      
\end{minipage}
The universal property of the sealing monad can be phrased in terms of a unique extension property: for any map $A \to B$ into a purely intensional type $B$, there is a unique extension $\P \vee A \to B$ along the unit of the sealing monad $\eta_{\phase \vee -} : A \to \P \vee A$. 
One can use the phase proposition $\P$ to exhibit a form of noninterference in the sense of information flow: any map \(A  \to  ( \P  \to  B)\) from a purely intensional type \(A\) to the extensional part of \(B\) must be constant. Therefore in \textbf{calf} one achieves a semantically faithful instrumentation of cost by programming against the writer monad on a \emph{purely intensional} cost monoid.

From an external point of view, \textbf{calf}'s types can be modeled as presheaves on the interval \( \I = { \left \{ \mathsf{ext}   \le   \mathsf{int} \right \} }\); in the terminology of Sterling and Harper~\cite{sterling-harper:2022}, such a presheaf encodes a set varying in the intension-extension ``security poset'' $\I$ whose presheaf restriction action \emph{redacts} intensional cost structure. The phase proposition $\P$ is given by the ``intermediate'' proposition $0 \to 1$ corresponding to the representable $\y{\I}(\kw{ext})$. 

\subsubsection{Synthetic domain theory}\label{subsubsec:SDT}

Synthetic domain theory (SDT) as a field started when Dana Scott conjectured that one ought to be able to reason about domains as if they were just sets \emph{provided} that one employs a constructive ambient metalanguage. In technical terms, this metalanguage can be construed as the internal languages of a topos, \ie an extensional dependent type theory. In general it is difficult to reconcile the domain-theoretic structure with the rich logical structure of dependent type theory, so the quest for SDT was in essence about constructing full subcategories of topoi that supported domain-theoretic constructions. Fullness is a critical property --- it means that \emph{every} map definable in type-theoretic language is a domain morphism, which absolves one from checking onerous side conditions when working internally. There are two well-known ways to obtain models of synthetic domain theory: one based on realizability (Hyland~\cite{hyland:1991}, Phoa~\cite{phoa:1991}, Reus~\cite{reus:1995}) and one based on sheaf topoi (Rosolini~\cite{rosolini:1986}, Fiore and Rosolini~\cite{fiore-rosolini:1997}, Fiore and Plotkin~\cite{fiore-plotkin:1996}, Matache \etal~\cite{matache-moss-staton:2021}). In this paper we will construct a sheaf model of SDT, so we take a moment to recall the basic ideas, mainly drawing from the work of Fiore and Plotkin~\cite{fiore-plotkin:1996}. 

\emph{Sheaf models of SDT.} The intuition behind (pre)sheaf models of synthetic domain theory is to use the Yoneda embedding to embed a category of predomains as a full subcategory of an ambient category whose internal language is dependent type theory. Thus given a (small\footnote{The fact that natural categories of predomains are rarely small can be overcome by means of Grothendieck universes or a small dense subcategory (\eg $\{\omega\} \xhookrightarrow{\text{full}} \OCPO{}$).}) category of predomains $G$, the category of presheaves on $G$ presents the original category of predomains as a full subcategory, which constitutes a simple model of SDT. However there are some strange properties of this model from an internal perspective. For example, because the $\emptyset$ is the initial object in $G$, it is a sub-predomain of every other predomain. But because the embedding does not preserve colimits, the subcategory of predomains does not contain an initial object from the perspective of \PSH{G}; in other words, this means that while the empty set is a subset of every predomain, it is not a sub-predomain, which obstructs the use of the ambient logic to reason about (pre)domains. We may repair the model by imposing a nontrivial Grothendieck topology on $G$ that ensures that the empty predomain is preserved by the embedding into the category of sheaves on $G$. In practice one may choose to preserve additional colimits as the application demands --- for instance in \cref{sec:model} we will construct a model of SDT in which discrete finite predomains are preserved.

\subsubsection{Predomains and the intrinsic order}\label{subsubsec:intrinsic-preorder}

Contrary to classical domain theory, ordering on predomains is a derived notion in synthetic domain theory. Nonetheless one may equip a synthetic predomain with an \emph{intrinsic preorder} that is analogous to the information order of ordinary predomains. By default the intrinsic/synthetic preorder is not very well-behaved: it is not pointwise on limits of predomains, and it is not even a partial order in general. In this paper we shall define and base our constructions on a notion of synthetic predomains for which the intrinsic preorder is extremely well-behaved: it is pointwise, partially ordered, and closed under \emph{synthetic} $\omega$-chains, which is the counterpart to closure under $\Nat$-chains for classic $\omega$-cpos.

Aside from intrinsic interest, we are motivated by two practical incentives for considering and developing the theory of the synthetic preorder. First, we would like to connect our work to the type theory for cost analysis developed by Grodin~\etal~\cite{grodin-niu-sterling-harper:2024}. In that setting one works with a \emph{cost ordering} on programs in which $e \le e'$ represents a proof of the fact that $e'$ is an upper bound of $e$ with respect to computational cost. We expect that one may extend the work of \emph{op. cit.} to account for recursive programs by the methods we develop in this paper, and the information/intrinsic preorder appears to be the correct theoretical framework for describing the interaction of cost and partiality~\cite[Section 8]{kavvos-morehouse-licata-danner:2019}. Second, the advantages of the synthetic point of view outlined \cref{subsubsec:SDT} do not detract from the benefits of the \emph{language} of classical domain theory. The intrinsic preorder integrates the advantages of both the synthetic and classical perspective by providing an intuitively appealing language for reasoning about domains that automatically complies with well-definedness conditions such as monotonicity and continuity.

\subsubsection{Synthetic predomains from orthogonality}\label{subsubsec:orthogonality}

  From an internal perspective, predomains in a model of SDT~\cite{fiore-plotkin:1996,reus-streicher:1999,sterling-harper:2022} are often defined in terms of (internal) orthogonality conditions. Intuitively, an orthogonality condition can be thought of as a way of specifying a closure property with respect to a given \emph{figure shape} $X \to Y$. A type \(A\) is \emph{orthogonal} to a map \(f : X  \to  Y\) when \(A^f : A^Y  \to  A^X\) is invertible. In other words \(A\) is orthogonal to \(f\) when \(A\) ``thinks'' \(f\) is an isomorphism. For instance, an $\omega$-cpo is a poset $P$ that is orthogonal to the figure shape $\{0 \le 1 \le \dots\} \hookrightarrow \{0 \le 1 \le \dots \le \infty\} $, \ie when $P$ is closed under joins of $\Nat$-chains. In \cref{sec:predomains} we define predomains by means of multiple such orthogonality conditions. The benefit of this general approach to synthetic domain theory is that subcategories defined by (internal) orthogonality conditions are automatically (internally) reflective --- this closes the predomains under structures necessary for day-to-day denotational semantics such as product and function types. Moreover, \emph{limits} of predomains are preserved by the inclusion into the ambient SDT topos, which allows one to reason about the denotational semantics in a straightforward way using the internal language. Note that colimits of predomains always exist by virtue of the reflection, but they are usually not preserved by inclusion; in some situations this may be calibrated by ensuring a precise correspondence between the original domain-theoretic site and the subcategory of predomains (\eg the internal characterization of $\omega$-cpos of Fiore and Rosolini~\cite{fiore-rosolini:1997:cpos}).

\subsection{Contributions}
The contribution of our work is summarized as follows:
\begin{enumerate}
  \item An axiomatization of a type theory based on SDT that incorporates both a phase distinction and a subuniverse of predomains whose intrinsic order structure is well-behaved (\cref{sec:type-theory}).
  \item A denotational semantics for \pcfc{} exhibiting noninterference between cost and behavior (\cref{sec:den-sem}).
  \item A dynamic semantics for \pcfc{} in which execution is modeled directly as computation (\cref{sec:comp-sem}).
  \item An internal cost-sensitive adequacy theorem identifying the two semantics at base type (\cref{sec:adequacy}).
  \item A relative sheaf model of SDT justifying the axioms of the type theory (\cref{sec:model}).
\end{enumerate}
\section{A type theory for cost-sensitive synthetic domain theory}\label{sec:type-theory}

We work in an extensional dependent type theory combining synthetic domain theory with the intension-extension phase distinction of Niu~\etal~\cite{niu-sterling-grodin-harper:2022}; the semantics and model construction of this theory is developed in \cref{sec:predomains} and \cref{sec:model}. Here we outline the language furnished by this combination. Following Niu~\etal~\cite{niu-sterling-grodin-harper:2022}, we assume an indeterminate proposition $\phase{}$ representing the phase distinction. We assume a (reflective) subuniverse of predomains $\mathcal{U}_\kw{predom}$. Every predomain is equipped with a \emph{synthetic} $\omega$-complete partial order structure, which we explain in \cref{subsec:synth-omega-cpo}. The significance of {synthetic} $\omega$-completeness is that they can be used to construct the join of \emph{rational chains} ($\Nat$-chains arising from the iterates of an endomap), the critical component in the interpretation of fixed-points in both classical and synthetic domain theory. 

By a \emph{domain} we shall mean a predomain $X$ equipped with a least element. Equivalently this may be characterized as an algebra for the \emph{lifting monad} $\mathbb{L} = (\kw{L}, \eta, \mu)$, which freely adjoins a least element to a predomain. There is a special domain $\Sigma \cong \lift{1} \hookrightarrow \Omega$ spanned by propositions called the \emph{Sierpiński space} or \emph{dominance}; propositions in $\Sigma$ should be thought of as the support of partial maps, and a map $X \to \Sigma$ can be thought of as a ``computational'' or ``observable'' subset of a predomain $X$. Moreover, because we use the \emph{sealing modality} (see \cref{subsubsec:phase-distinction}) associated with the phase proposition $\P$ to define the denotational cost semantics of \pcf{}, we require that predomains are closed under the sealing modality. Lastly, we assume that the subuniverse $\mathcal{U}_\kw{predom}$ is closed under lifting.

\subsection{Denotational semantics of cost-sensitive \pcf{}}\label{subsec:den-sem}

We aim to prove the computational adequacy property for a cost-sensitive version of \pcf{} called \pcfc{} in which cost and partiality are treated as a single \emph{call-by-push-value effect}. We spell out the details of the language in \cref{sec:pcfc}, but the basic idea is to separate pure values and effectful computations at both the term and type level, resulting in a class of value types $A$ and a class of computation types $X$. The type structure of the language is generated by a pair of type operators $\kw{F}, \kw{U}$ in which $\kw{F}(A)$ represents the partial cost-aware computations of type $A$ and $\kw{U}(X)$ represents the value type whose underlying set of points are computations of type $X$. In contrast to call-by-value languages, call-by-push-value function types are \emph{computations} and take the form $A \to X$. As an example, the ordinary call-by-value \pcf{} function type $\kw{nat} \to \kw{nat}$ corresponds to the type $\kw{U}(\kw{nat} \to \kw{F}(\kw{nat}))$ in the call-by-push-value setting --- here the codomain type $\kw{F}(\kw{nat})$ records both the cost incurred and the possibility of divergence, and the outer $\kw{U}(-)$ corresponds to the fact that functions are values in a call-by-value setting. 

The denotational semantics of \pcfc{} is based on the adjunction models of call-by-push-value~\cite{levy:2003:book} in which $\sem{\kw{F}(A)}$ is defined to be the lift of $\mathbb{C} \times \sem{A}$ for a \emph{purely intensional} cost monoid $\mathbb{C}$ in the sense of \cref{subsubsec:phase-distinction}. By using a purely intensional type for the cost semantics, we may prove a cost-sensitive computational adequacy theorem that can be restricted to an extensional adequacy result for \pcf{}. Moreover, by interpreting the cost this way, the resulting denotational cost semantics evinces a form of \emph{information flow security} when viewing cost and behavior as security levels:

\begin{proposition}
  For any purely extensional type $B$, every map $f : \mathbb{C} \to \lift{B}$ is \emph{weakly constant} in the sense that for any inputs $x, y : \mathbb{C}$, we have that $f~x$ and $f~y$ are equal whenever they are both defined. 
\end{proposition} 

As an application of computational adequacy, we can immediately transfer this intrinsic denotational security property to functions of \pcfc{}. 

\begin{remark}
  As mentioned in \cref{subsubsec:internal}, we will define both the syntax and semantics of \pcfc{} as objects and functions in the internal type theory of a synthetic domain theory topos. In a more traditional approach to denotational cost semantics, one may instead define a model of \pcfc{} in the category $\omega\textbf{CPO}^\to$ of families of $\omega$-cpos, which is semantically simpler in comparison to models of synthetic domain theory. In this paper we pursue an approach based on the latter in order to unify two strands of prior work --- cost-sensitive programming and verification in dependent type theory~\cite{niu-sterling-grodin-harper:2022} and internal denotational semantics~\cite{paviotti-mogelberg-birkedal:2015,niu-harper:2023} --- and synthetic domain theory provides the means for smoothly integrating higher-order recursion into dependent type theory. 
\end{remark}

\section{Cost-sensitive predomains in synthetic domain theory}\label{sec:predomains}

As mentioned in the \cref{subsubsec:intrinsic-preorder}, the purpose of this section is to define a notion of predomains in synthetic domain theory such that the intrinsic preorder on predomains is partially ordered, defined pointwise on function spaces, and is closed under suprema of synthetic \(\omega\)-chains. We will use these properties to define and reason about \emph{admissible subsets of domains}. We define the basic notions of synthetic domain theory, and give an axiomatization of the rest of the paper in terms of \emph{SDT models of the intension-extension phase distinction}, culminating in a category of predomains satisfying the properties laid out above. 

\subsection{Partial maps, dominance, and lifting}

In a category with pullbacks, a \emph{partial map} $A \rightharpoonup B$ is a span $A \hookleftarrow D \to B$ consisting of a subset $D \hookrightarrow A$ on which $A \rightharpoonup B$ is defined. In synthetic domain theory, only certain monomorphisms correspond to domains of definitions of partial maps. In the terminology of Rosolini~\cite{rosolini:1986} such a collection is called a dominion:

\begin{definition}
  A pullback-stable collection of monomorphisms is a called a \emph{dominion} when it is closed under identity and composition. 
\end{definition}

\begin{definition}
  Let \CatIdent{E} be an elementary topos equipped with a subobject \(\Sigma   \hookrightarrow  \Omega\) such that $\top \in \Sigma$. A monomorphism is \emph{classified} by $\Sigma$ if and only if its characteristic map factors through $\Sigma \hookrightarrow \Omega$. A \emph{dominance} is a subobject $\Sigma \hookrightarrow \Omega$ such that the class of monos classified by $\Sigma$ is a dominion. We call a proposition (resp., predicate) factoring through $\Sigma$ a \emph{$\Sigma$-proposition} (resp., \emph{$\Sigma$-predicate}). 
\end{definition}

In the internal language, this means that \(\Sigma\) contains the true proposition \(\top\) and is closed under dependent sums, which we write as \(\phi   \mathbin{\angle}  f\) given \(\phi  :  \Sigma\) and \(f :  \phi   \to   \Sigma\). These structures ensure that the dominance determines a \emph{lifting monad} $\mathbb{L} = (\kw{L}, \eta, \mu)$ whose action on points are defined as follows:

\begin{definition}
  The \emph{lift} of a type $A$ relative to a dominance $\Sigma$ is defined as $\lift{A} = \Sigma_{\phi : \Sigma}.~\phi \to A$. 
\end{definition}

The lifting monad is also called the \emph{partial map classifier} because every partial map $A \hookleftarrow D \to B$ with $D$ a $\Sigma$-subobject of $A$ appears as the pullback of $\eta_B : B \to \lift{B}$ for a unique $A \to \lift{B}$. 

\begin{nota}
  Given a partial element \(e :  \mathsf{L} A\), we write \(e { \downarrow }  :  \Sigma\) for its support, \ie{} $- {\downarrow}$ is the first projection $\kw{L}(A) \to \Sigma$. When it is known that \(e { \downarrow }\) holds, we may write \(e : A\) for the defined element. 
\end{nota}

\subsection{Complete types}

In synthetic domain theory the structure of predomains is generated not from consideration of $\Nat$-indexed chains but rather a new notion of chains called a \emph{synthetic} $\omega$-chain, which is defined simply to be a map out of $\omega$. The difference between the two notions of chain is  elucidated by the fact that $\Nat$ is the initial lifting algebra for the dominance of decidable propositions, whereas $\omega$ is the initial algebra for a larger dominance.

\begin{definition}
A \emph{synthetic \(\omega\)-chain} is a map \(\omega \to  A\) from the initial $\kw{L}$-algebra or lift algebra $\omega$.
\end{definition}

The closure properties of synthetic $\omega$-chains can be captured by considering an orthogonality condition relative to the figure shape \(\omega   \hookrightarrow   \overline{\omega}\) induced by the inclusion of the initial lift algebra in the \emph{final lift coalgebra} $\overline{\omega}$. As we alluded to in \cref{subsubsec:orthogonality}, orthogonality is a way to identify types that ``think'' certain maps are isomorphisms: 

\begin{definition}
A type \(A\) is \emph{orthogonal} to \(f : X  \to  Y\) when there is a unique extension of any map \(g : X  \to  A\) to a map \(\overline{g}  : Y  \to  A\) such that \(g =  \overline{g}   \circ  f\). 
\end{definition}

\begin{definition}
  A type \(A\) is called \emph{complete} when it is orthogonal to the figure shape $\omega \hookrightarrow \overline{\omega}$, and \emph{well-complete} when $\lift{A}$ is complete.
\end{definition}

Complete types~\cite{hyland:1991,phoa:1991} are the synthetic counterpart to $\omega$-cpos in classic domain theory. The class of well complete types was introduced in Longley and Simpson~\cite{longley-simpson:1997} as the least restrictive possible notion of predomain that is closed under lifting. In this paper we consider the dual \emph{most restrictive} class of predomain, the \emph{replete} types~\cite{hyland:1991}, in order to obtain a sharper characterisation of the intrinsic order relation (see \cref{subsec:intrinsic}) on predomains. 

\begin{definition}
  A map \(f : X  \to  Y\) is called \emph{\(\Sigma\)-equable} or a \emph{\(\Sigma\)-isomorphism} when \(\Sigma\) is orthogonal to \(f\). 
\end{definition}

\begin{definition}\label[defn]{def:replete}
  A type is \emph{replete} when it is orthogonal to every \(\Sigma\)-iso. A \emph{predomain} is a replete type. 
\end{definition}

In other words, a predomain respects every $\Sigma$-isomorphism; we will use this fact to easily transfer properties that hold of $\Sigma$ to every predomain in \cref{sec:properties}.

\subsection{The intrinsic order}\label{subsec:intrinsic}

For every type $A$, the dominance $\Sigma$ induces an \emph{intrinsic preorder} $\specle_A$ on $A$ analogous to the specialization preorder on a topological space: $x \specle_A y$ if and only if $f~x$ implies $f~y$ for every \(f:A \to \Sigma\). Viewing $f : A \to \Sigma$ as a computational or observable property of $A$, $x \specle_A y$ holds whenever $y$ satisfies every observable property of $x$. By default the intrinsic preorder is relatively unconstrained --- for instance, it need not be a partial order in general and on the dominance $\Sigma$ it need not coincide with the implication order on propositions. The purpose of this section is to axiomatize some constraints on $\Sigma$ that will make intrinsic preorder coincide with the implication order on $\Sigma$; this is used in the characterization of the intrinsic order of lifting in \cref{sec:properties}. To this end, we introduce an intermediate \emph{path relation} on types: 

\begin{definition}
  A \emph{path} in a type $A$ is a map $\Sigma \to A$. The \emph{boundary} of a path $f:\Sigma\to A$ is the pair \(\boundary f = (f~ \bot , f~ \top )\). We write $x \pathle_A y$ when there exists a path whose boundary is $(x, y)$. 
\end{definition}

The path relation is an alternative way to surface the order structure of predomains, studied in much more detail by Fiore~\cite{fiore:1995}; the path relation is also used by Grodin~\etal~\cite{grodin-niu-sterling-harper:2024} to obtain a theory of synthetic preorders for cost analysis. In this paper one can view the path relation as an auxiliary notion that ultimately coincides with the intrinsic preorder on predomains. In general we call types for which this holds linked, following Phoa~\cite{phoa:1991}:

\begin{definition}
 A type is called \emph{linked} when the intrinsic preorder coincides with the path relation.
\end{definition}

The fact that the intrinsic order on $\Sigma$ coincides with the implication order follows if $\Sigma$ is linked. This latter property holds when $\Sigma$ satisfies one of the fundamental axioms of synthetic domain theory: 
\begin{definition}
  The dominance $\Sigma$ satisfies \emph{Phoa's principle} when the boundary evaluation map $\boundary : \Sigma^\Sigma \to \Sigma \times \Sigma$ factors as an isomorphism followed by an inclusion: $\Sigma^\Sigma \cong \{(\phi, \psi) \mid \phi \to \psi\} \hookrightarrow \Sigma \times \Sigma$. 
\end{definition}

Phoa's principle expresses the fact that the path space of $\Sigma$ is fully characterized by ordered pairs of $\Sigma$-propositions (with respect to implication). Because the negation of an observable property is not in general observable, Phoa's principle may be seen as an explicit statement of the constructive/observable nature of $\Sigma$-propositions, \ie there is no map $\Sigma \to \Sigma$ sending $\phi$ to $\neg\phi$. 

\begin{restatable}{proposition}{PropDomLinked}\label[prop]{prop:dom-linked}
Assuming Phoa's principle, the dominance \(\Sigma\) is linked.
\end{restatable}

\begin{corollary}
  Assuming Phoa's principle, the intrinsic order on \(\Sigma\) is the implication order.
\end{corollary}

Moreover, the path in $\Sigma$ associated to every implication $\phi \to \psi$ is unique in the following sense: 
\begin{definition}
  A type $A$ is \emph{boundary separated} when any two paths in $A$ sharing a boundary are equal.
\end{definition}

\begin{proposition}\label[prop]{prop:sigma-boundary-sep}
  Assuming Phoa's principle, $\Sigma$ is boundary separated. 
\end{proposition}

\subsection{Axioms of the phase distinction in SDT and predomains}\label{subsec:predomains} 

Having developed the axiomatics of the ordinary synthetic domain theoretic components of our work, we now introduce the intension-extension phase distinction as discussed in \cref{subsubsec:phase-distinction}. The meeting point of the domain-theoretic and the cost-sensitive aspects of our work is simply expressed by the requirement that the phase proposition $\phase$ is a $\Sigma$-proposition. This allows us to manufacture a \emph{purely intensional} monoid from a standard cost monoid that we will use to define a denotational semantics for \pcf{} that exhibits the natural information flow security properties with respect to cost and behavior as outlined in \cref{subsec:den-sem}. 

More specifically, we may use the fact that $\phase$ is a $\Sigma$-proposition to define the sealing monad $\phase \vee -$ mentioned in \cref{subsubsec:phase-distinction} that can be seen as the canonical way of making a predomain purely intensional. When defining the denotational semantics, we may then take as an input to the model construction any ordinary monoid (in the subuniverse of predomains) and apply the sealing monad to obtain a purely intensional monoid. Lastly, to exhibit the kind of noninterference property of cost and behavior as discussed in \cref{subsubsec:phase-distinction}, we require that there is a purely extensional base predomain. These considerations lead us to the following axiomatization of our synthetic domain theory topos: 

\begin{definition}\label[defn]{def:axioms}
  An \emph{SDT model of the intension-extension phase distinction} is an elementary topos \CatIdent{E} equipped with a dominance $\Sigma$ satisfying Phoa's principle such that $\Sigma$ is complete, and a distinguished $\Sigma$-proposition $\P$ such that the type of booleans $2$ is an extensional predomain. 
\end{definition}

For the rest of the paper we assume a given SDT model of the intension-extension phase distinction \CatIdent{D}. All constructions are all carried out in the internal language of \CatIdent{E}.

\section{Properties of predomains}\label{sec:properties}

We now establish the expected characterizations of the intrinsic preorder (\cref{subsec:intrinsic}) on predomains:
\begin{enumerate}
  \item The intrinsic preorder is pointwise on products, functions, and liftings of predomains. 
  \item The intrinsic preorder on a predomain is a \emph{synthetic $\omega$-complete partial order}. 
  \item The synthetic $\omega$-complete partial order structure of predomains is defined componentwise for products and functions between predomains. 
\end{enumerate} 

The general properties of the intrinsic preorder and link relation in SDT have been investigated in several prior works~\cite{phoa:1991,reus:1995,longley-simpson:1997}. Here we recall just a few important properties that we will need.

\begin{restatable}{proposition}{PropPredomainProperties}\label[prop]{prop:predomain-properties}
  Any predomain $A$ enjoys the following properties:
  \begin{enumerate}
    \item{\emph{Completeness}: \(A\) is orthogonal to \(\omega   \hookrightarrow   \overline{\omega}\).}  
    \item{\emph{Anti-symmetry}: the intrinsic preorder on \(A\) is a partial order.}  
    \item{\emph{Boundary separation}: maps \(\Sigma   \to  A\) with equal boundary are equal.}  
    \item{\emph{Linkedness}: the intrinsic preorder and the link relation on \(A\) coincide.}
  \end{enumerate}    
\end{restatable}

In addition, because predomains are defined in terms of orthogonality conditions, they are closed under (internal) limits and have all colimits, \ie{} limits of predomains are computed in the same way as limits of general types. The rest of the section is dedicated to proving the desired properties on the intrinsic order on predomains. Because the intrinsic order $\specle$ and $\pathle$ are the same for predomains, we may speak of the \emph{synthetic order} of predomains and write $\le$ for this relation in the rest of the paper. 

\subsection{Discrete predomains}\label{subsec:disc-predom}

We shall assume that the cost structure of \pcfc{} is given as a discrete type in the following sense: 
\begin{definition}\label[defn]{def:discrete}
A type \(A\) is called \emph{flat} or \emph{discrete} when \(x  \le^\circ  y\) implies \(x = y\).
\end{definition}

\begin{definition}
  A type \emph{has \(\Sigma\)-equality} when its equality relation is valued in \(\Sigma\)-propositions.
\end{definition}

\begin{proposition}\label[prop]{prop:sigma-dec-disc}
  Any type with $\Sigma$-equality is discrete. 
\end{proposition}

\begin{proof}
  Let \(f : A  \to   \Sigma\) be the characteristic map that sends \(a\) to \(a = x\), which by assumption is a \(\Sigma\)-proposition. Since \(x  \sqsubseteq _A y\) on the specialization order and \(f(x)\) holds, we have that \(f(y)\) holds as well. 
\end{proof}

The category of predomains possesses a natural numbers type $\NatP$ with $\Sigma$-equality, which means it is also discrete. Note that it is not necessarily the same as the ambient natural numbers type $\Nat$, and we will not assume that it is the case in our constructions. From a logical perspective, this difference means that $\NatP$ has a universal mapping-out property whose motive is valued in predomains rather than arbitrary types. 

In \cref{sec:den-sem} we will require the cost structure of \pcfc{} to be both discrete and purely intensional in the sense of \cref{subsubsec:phase-distinction}; here we show that these are compatible requirements by giving sufficient conditions to obtain a purely intensional discrete type (for instance, $\P \vee \Nat$ will be discrete).  

\begin{restatable}{proposition}{PropSealingSigDecidable}\label[prop]{prop:sealing-sig-decidable}
  If $A$ has $\Sigma$-equality, then so does $\P \vee A$. 
\end{restatable}

\subsection{Characterization of the synthetic order} 

As mentioned at the beginning of this section, one of the primary motives of using replete types as predomains is to give a compositional characterization of the synthetic order. Roughly this means that the order relation on composite predomains can be defined in terms of the order on the constituent predomains.

\begin{restatable}{proposition}{PropFuncPointwise}
  When $A$ and $B$ are predomains, the synthetic orders on \(A  \times  B\) and \(A \to  B\) are pointwise.
\end{restatable}

We may also give a similar characterization of the order relation on lifted predomains: 

\begin{restatable}{proposition}{PropLiftPointwise}\label{prop:lift-order}
  Given a predomain \(A\), we have that \(x  \sqsubseteq _{ \mathsf{L} A } y\) if and only if \(x{ \downarrow }\) implies \(y{ \downarrow }\) and whenever \(x{ \downarrow }\), we have \(x  \sqsubseteq _A y\). 
\end{restatable}

\subsection{The synthetic \texorpdfstring{$\omega$}{omega}-complete partial order structure}\label{subsec:synth-omega-cpo}

Every predomain $A$ contains the least upper bound of a synthetic $\omega$-chain; we will use this fact to interpret fixed points of \pcfc{} in \cref{sec:den-sem}. 

\begin{restatable}{proposition}{PropLub}\label[prop]{prop:lub}
  For every map \(f: \omega   \to  A\) into a complete type \(A\), there exists an element \(a_ \infty  : A\) such that \(a_ \infty\) is a least upper bound of \(f\) with respect to the intrinsic preorder. 
\end{restatable}

Because the intrinsic preorder on a predomain is a partial order by \cref{prop:predomain-properties}, we write $\sup f$ for the (necessarily unique) element defined in \cref{prop:lub}. We note that suprema of synthetic $\omega$-chains in function spaces are computed pointwise.

\subsection{Domains and admissibility}

Semantically, recursive functions may be interpreted as the fixed points of endomaps of domains: 

\begin{definition}
  A \emph{domain} is a predomain equipped with a \(\mathsf{L}\)-algebra structure. Every domain \(D\) contains a least element given by postcomposing the algebra map with the unique undefined element \(( \bot , !) : 1  \to   \mathsf{L} D\). We write \(\bot _D : D\) for the least element of \(D\). 
\end{definition}

\begin{proposition}
Given a map of domains $f : D \to D$, there is an element $\kw{fix}(f) : D$ such that $\kw{fix}(f)$ is the least fixed-point of $f$. 
\end{proposition}

Similar to classical domains, we may introduce a notion of ``good subsets'' of domains for which fixed-point induction is valid.

\begin{definition}\label[defn]{def:admissible}
  A subset of a domain \(D\) is \emph{admissible} when it is complete and contains \(\bot _D\). 
\end{definition}

\begin{remark}
  It might be tempting to define admissible subsets as subsets of domains that are closed under the least element and synthetic $\omega$-joins, but it is unclear that this definition actually supports fixed-point induction in general. However this intuitive definition does suffice for all instances of fixed-point induction used in this paper. 
\end{remark}

\begin{restatable}{proposition}{BugAdmissible}\label[prop]{prop:bug-admissible}
  A \emph{lower} subset of a complete type is complete when it is closed under synthetic $\omega$-joins. 
\end{restatable}

\begin{corollary}\label[cor]{cor:admissible}
  A lower subset of a domain is admissible when it is closed under the least element and synthetic $\omega$-joins.
\end{corollary}

\begin{proposition}[Fixed-point induction]
  Given an admissible subset $A \subseteq D$ of a domain $D$ and $f : D \to D$, to show that $\kw{fix}(f) \in A$ it suffices to show that $x \in A$ implies $f~x \in A$. 
\end{proposition}

\subsection{Monotonicity and continuity}

As we discussed in \cref{subsubsec:SDT}, one of the main benefits of working in a synthetic domain theory is that maps are automatically compatible with the derived order structure: 

\begin{restatable}{proposition}{PropMonotone}\label[prop]{prop:monotone}
  Every map $f : A \to B$ between predomains is monotone in the synthetic order. 
\end{restatable}

\begin{restatable}{proposition}{PropCont}\label[prop]{prop:cont}
  Every map $f : A \to B$ between predomains is continuous in the sense that $f(\sup d) = \sup(f \circ d)$ for every synthetic $\omega$-chain $d : \omega \to A$. 
\end{restatable}

\section{\texorpdfstring{\pcfc{}}{pcf/cost}: a language for cost-aware higher-order recursion}\label{sec:pcfc}

Our main technical result is the computational adequacy property for \pcfc{}, a version of Plotkin's \pcf{}~\cite{plotkin:1977} equipped with an abstract cost effect. The treatment of both cost structure and recursion as a call-by-push-value effect in \pcfc{} is inspired by~\cite{kavvos-morehouse-licata-danner:2019}. The syntax of \pcfc{} is parameterized in a monoid $\mathbb{C}$ representing the cost structure; in \cref{sec:den-sem} we will impose further properties on $\mathbb{C}$ when we define the denotational semantics of \pcfc{}. As mentioned in \cref{subsec:den-sem}, the type structure of \pcfc{} is generated by a pair of operators $\kw{F}, \kw{U}$ that corresponds semantically to the free-forgetful adjunction between plain sets and the category of \emph{partial cost algebras}, which are sets equipped with an action for the partial cost monad $\lift{(\mathbb{C} \times -)}$. The types and terms of \pcfc{} are summarized in \cref{fig:pcfc}, defined as inductive definitions in the SDT topos \CatIdent{E}. 

\begin{figure}
\begingroup
\setlength\columnsep{-6.5cm}
\begin{multicols}{2}
  \iblock{
    \mhang{
      \textbf{inductive}~\kw{Ty}^+ : \kw{Set}~\textbf{where}
    }{
      \mrow{\kw{ans} : \kw{Ty}^+}
      \mrow{\kw{nat} : \kw{Ty}^+}
      \mrow{\kw{U} : \kw{Ty}^\ominus \to \kw{Ty}^+}
    }
    \row
    \mhang{
      \textbf{inductive}~\kw{Ty}^\ominus : \tpv~\textbf{where}
    }{
      \mrow{\kw{F} : \kw{Ty}^+ \to \kw{Ty}^\ominus}
      \mrow{{\rightharpoonup} :\kw{Ty}^+ \to \kw{Ty}^\ominus \to \kw{Ty}^\ominus}
    }
    \row 
    \mrow{
      \kw{Tm}^\ominus : \kw{Con} \to \kw{Ty}^\ominus \to \tpv
    }
    \mrow{
      \kw{Tm}^\ominus(\Gamma, X) = \kw{Tm}^+(\Gamma, \kw{U}(X))
    }
  }
  \columnbreak 
  \iblock{
    \mhang{
      \textbf{inductive}~\kw{Tm}^+ : \kw{Con} \to \kw{Ty}^+ \to \kw{Set}~\textbf{where}
    }{
      \mrow{\kw{var} : \kw{Var}(\Gamma, A) \to \kw{Tm}(\Gamma, A)}
      \mrow{\kw{yes} : \kw{Tm}^+(\Gamma, \kw{ans})}
      \mrow{\kw{no} : \kw{Tm}^+(\Gamma, \kw{ans})}
      \mrow{\kw{zero} : \kw{Tm}^+(\Gamma, \kw{nat})}
      \mrow{\kw{succ} : \kw{Tm}^+(\Gamma, \kw{nat}) \to \kw{Tm}^+(\Gamma, \kw{nat})}
      \mrow{\kw{ap} : \kw{Tm}^\ominus(\Gamma, A \rightharpoonup X) \to \kw{Tm}^+(\Gamma, A) \to \kw{Tm}^\ominus(X)}
      \mrow{\kw{ret} : \kw{Tm}^+(\Gamma, A) \to \kw{Tm}^\ominus(\Gamma, \kw{F}(A))}
      \mrow{\kw{step} : \mathbb{C} \to \kw{Tm}^\ominus(\Gamma, X) \to \kw{Tm}^\ominus(\Gamma, X)}
      \mrow{\kw{bind} : \kw{Tm}^\ominus(\Gamma, \kw{F}(A)) \to \kw{Tm}^\ominus(A::\Gamma, X) \to \kw{Tm}^\ominus(\Gamma, X)}
      \mrow{\kw{ifz} : \kw{Tm}^+(\Gamma, \kw{nat}) \to \kw{Tm}^\ominus(\Gamma, X) \to \kw{Tm}^\ominus(\kw{nat}::\Gamma, X) \to \kw{Tm}^\ominus(\Gamma, X)}
      \mrow{\kw{fix} : \kw{Tm}^\ominus(\kw{U}(X)::\Gamma, X) \to \kw{Tm}^\ominus(\Gamma, X)}
      \mrow{\kw{lam} : \kw{Tm}^\ominus(A::\Gamma, X) \to \kw{Tm}^\ominus(\Gamma, A \rightharpoonup X)}
    }
  }
\end{multicols}
\endgroup
\caption{The grammar of types and terms in \pcfc. We will often omit $\kw{Tm}^{\{+,\ominus\}}$ in the case of closed terms and write $\Gamma \vdash \{A, X\}$ for the type $\kw{Tm}^{\{+, \ominus\}}(\Gamma, \{A, X\})$. Given $e : \U{\F{A}}$ and $f : \U{(A \to X)}$, we also write $e; f$ for $\kw{bind}(e, f) : \U{X}$. }
\label{fig:pcfc}
\end{figure}

\section{Cost-aware denotational semantics of \texorpdfstring{\pcfc{}}{pcf/cost}}\label{sec:den-sem}

Both the syntax of \pcfc{} and our model construction is parameterized in a monoid object \(( \mathbb{C} , +, 0)\) representing the cost structure. We require the following conditions to hold: 
\begin{enumerate}
  \item{\emph{Computational}: \(\mathbb{C}\) is a predomain with \(\Sigma\)-equality. }  
  \item{\emph{Phase separation}: \(\mathbb{C}\) is purely intensional, \ie $(\P \to \mathbb{C}) \cong 1$. }  
\end{enumerate}

That $\mathbb{C}$ is computational are used in two places in the computational adequacy proof: when reasoning about the \emph{computational semantics} of \pcfc{} (see \cref{sec:comp-sem}), we need $\mathbb{C}$ to be discrete (which follows from \cref{prop:sigma-dec-disc}) in order to prove the property that sequential composition of computations may be decomposed (\cref{prop:comp-ind-seq}), and we need $\mathbb{C}$ to have $\Sigma$-equality when showing that the \emph{formal approximation predicates} (see \cref{fig:formal-approx-rel}) associated to semantic domains are admissible. We require $\mathbb{C}$ to be purely intensional to ensure that the denotational semantics of \pcfc{} exhibits an \emph{intrinsic} noninterference property of cost and behavior as sketched in \cref{subsec:den-sem}. As discussed in \cref{subsec:predomains}, there is a canonical way of turning any monoid $\mathcal{M}$ into a purely intensional type by means of the sealing monad $\P \vee -$. 

\subsection{The partial cost monad}

To model partiality and cost as a single effect, the computation types of \pcfc{} are interpreted as algebras for the monad $(\mathbb{T}, \eta_\mathbb{T}, \mu_\mathbb{T})$ whose action on points is defined by composing the lift and the writer monad $\kw{T} A =  \mathsf{L} ( \mathbb{C}   \times  A)$. The distributive law for $\kw{T}$ and the resulting monad structure is displayed in \cref{fig:partial-cost-monad}, where we write $x \leftarrow_\mathbb{M} e; f(x)$ for the induced bind operation of a monad $\mathbb{M}$ where $f : A \to \kw{M}(B)$ is a map into a free $\mathbb{M}$-algebra. We will also write $f^\sharp(e)$ for sequencing $e : \kw{M}(A)$ and $f : A \to X$ for an $\mathbb{M}$-algebra $X$. 
\begin{figure}
  \begin{multicols}{2}
    \iblock{
    \mrow{\tau : \mathbb{C} \times \lift{A} \to \lift{(\mathbb{C} \times A)}}
    \mrow{\tau(c, ( \phi , f))  =  ( \phi ,  \lambda  u: \phi .~(c, f u))}
    }
  \columnbreak
  \iblock{
    \mrow{\eta_\mathbb{T}(a) = \eta_\kw{L}(0, a)}
    \mrow{\mu_\mathbb{T}(e) = (c, x) \leftarrow_\mathbb{L} e; (c', a) \leftarrow_\mathbb{L} x; (c + c', a)}
  }
  \end{multicols}
  \caption{Left: distributive law; right: monad structure of $\mathbb{T}$.}
  \label{fig:partial-cost-monad}
\end{figure}

\subsection{The derived cost algebra}

To model the cost effect $\kw{step} : \mathbb{C} \times X \to X$, we use the fact that every algebra for $\mathbb{T}$ is canonically an algebra for the writer monad $\mathbb{C} \times -$, which is a general property of composite monads defined from a distributive law~\cite[Section 2]{beck_distributive_1969}. In our case this means that every $\mathbb{T}$-algebra has an underlying cost algebra as well; we write $\costmap : \mathbb{C} \times X \to X$ for the cost algebra map. 

\begin{restatable}{proposition}{PropCostAlgLaws}\label[prop]{prop:cost-alg-laws}
  The action of the derived cost algebra satisfies the following equations for $e : \kw{T}(A)$ and $f : A \to X$ for some $\mathbb{T}$-algebra $X$: 
  \[\begin{aligned}   
    &c  \boxplus_{\kw{T}(A)} e = ( e { \downarrow } ,  \lambda  u.~ \mu _{ \mathbb{C}   \times  -}(c, e)) \\   
    &f^\sharp(c  \boxplus_{\kw{T}(A)}  e) = c  \boxplus_X  (f^\sharp~e) \\   
    &(c  \boxplus_{A \to X}  f)~a = c  \boxplus_{X}  (f~a) 
  \end{aligned}\]
\end{restatable}

\subsection{Denotational semantics of \texorpdfstring{\pcfc{}}{pcf/cost}}

The semantics of \pcfc{} is based around the free-forgetful adjunction associated with the partial cost monad in which value types are predomains and computation types are $\mathbb{T}$-algebra valued in predomains.
The essential parts of the model is displayed in \cref{fig:model}. Most of the type structure of \pcfc{} is interpreted using the cartesian closed structure of predomains; note that the numerals type is interpreted as the natural numbers type of predomains $\NatP$, which is not the same as the ambient natural numbers $\Nat$.  

\begin{figure}
\begingroup
\setlength\columnsep{-6.5cm}
\begin{multicols}{2}
  \iblock{
    \mrow{\sem{-} : \tpv \to \mathcal{U}_\kw{predom}}
    \mrow{\sem{-} : \tpc \to \kw{Alg}_\kw{\mathbb{T}}(\mathcal{U}_\kw{predom})}
    \row
    \mrow{\sem{\F{A}} = \kw{T}(\sem{A})}
    \mrow{\sem{\U{X}} = U(\sem{X})}
    \mrow{\sem{1} = 1}
    \mrow{\sem{\kw{ans}} = 2}
    \mrow{\sem{\kw{nat}} = \NatP}
    \mrow{\sem{A \rightharpoonup X} = \sem{A} \to \sem{X}}
  }
  \columnbreak
  \iblock{
    \mrow{\sem{-} : \impl{\Gamma, A} (\Gamma \vdash A) \to \sem{\Gamma} \to \sem{A}}
    \mrow{\sem{-} : \impl{\Gamma, A} (\Gamma \vdash X) \to \sem{\Gamma} \to U(\sem{X})}
    \row
    \mrow{\sem{\kw{ret}(a)}(\gamma) = \eta_\mathbb{T}(\sem{a}(\gamma))}
    \mrow{\sem{\kw{step}(c, e)}(\gamma) = c \costmap \sem{e}(\gamma)}
    \mrow{\sem{\kw{bind}(e, f)}(\gamma) = \kw{bind}(\sem{e}(\gamma), \lambda a.~\sem{f}(a, \gamma))}
    \mrow{\sem{\kw{fix}(f)}(\gamma) = \kw{fix}(\lambda x.~\sem{f}(x, \gamma))}
  }
\end{multicols}
\endgroup
\caption{Selected clauses of the model; in the above we write $U(X)$ for the carrier of a $\mathbb{T}$-algebra $X$.}
\label{fig:model}
\end{figure}

\section{Computational semantics of \texorpdfstring{\pcfc{}}{pcf/cost}}\label{sec:comp-sem}

Computational adequacy is a property relating the denotational semantics of a language with its execution behavior as given by an \emph{operational semantics}. Commonly one employs a \emph{structural operational semantics}~\cite{plotkin:2004} in which the operational semantics of a language is defined as an inductive family of relations ${\mapsto_A} \subseteq A \times A$. Both termination and evaluation can be defined by means of the reflexive transitive closure of the family of relations $\mapsto_A$, which is the smallest reflexive, transitive family of relations containing $\mapsto_A$. 

\emph{Computational semantics.} We consider an alternative formulation of the dynamic semantics of \pcfc{} in which execution is modeled not as an inductive family but directly as a (partial) computation. This departure from traditional operational semantics is necessitated by our proof of computational adequacy. As we explain in \cref{subsec:adequacy}, termination must be a predicate valued in $\Sigma$-propositions; however, the reflexive-transitive closure of a decidable relation (such as $\mapsto_A$) can only be seen to be a $\Sigma$-predicate assuming that $\Sigma$ is closed under countable joins of decidable propositions that are \emph{preserved} by the inclusion into the ambient type theory, \ie{} $\some{n:\Nat} \phi(n) \in \Sigma$ for every countable family of $\Sigma$-propositions $\phi$. Our approach bypasses this question by de-emphasizing $\Nat$ in favor of the initial lifting algebra $\omega$, a lesson of Simpson~\cite{simpson:2004}. 
We discuss the relationship between this computational semantics and ordinary operational semantics in \cref{sec:conclusion}. 

\subsection{Computational semantics of \texorpdfstring{\pcfc{}}{pcf/cost}}\label{subsec:comp-sem}

\begin{figure}
\begin{mathpar}
  \inferrule{
  }{
    \kw{bind}(\kw{ret}(a), f) \mapsto 0, f(a)
  }

  \inferrule{
  }{
    \kw{ap}(\kw{lam}(e), e_1) \mapsto 0, e(e_1)
  }

  \inferrule{
  }{
    \kw{fix}(e) \mapsto 0, e(\kw{fix}(e))
  }

  \inferrule{
  }{
    \kw{ifz}(\kw{zero}, e_0, e_1) \mapsto 0, e_0
  }

  \inferrule{
  }{
    \kw{ifz}(\kw{succ}(v), e_0, e_1) \mapsto 0, e_1(v)
  }

  \inferrule{
  }{
    \kw{step}^c(e) \mapsto c, e
  }
\end{mathpar}
\caption{Rules for the small-step transition relation.}
\label{fig:transitions}
\end{figure}
We begin with a family of small-step transition relations \({ \mapsto_A }  \subseteq   \tmv{A}  \times   \mathbb{C}   \times   \tmv{A}\) that implements the \emph{cost effect} model in the sense of Hoffmann~\cite{hoffmann:2011:thesis}. The intuitive meaning of \(e  \mapsto  c, e'\) is that \(e\) transitions in one step to \(e'\) and incurs cost \(c\); the only place where cost is effected is at \(\mathsf{step}\): $\mathsf{step} ^c(e)  \mapsto  c, e$. The family of relations $\mapsto_A$ is defined as an inductive family whose generators are displayed in \cref{fig:transitions} (the expected congruence rules have been omitted). We iterate the ``one step'' relation to obtain the partial map implementing the computational semantics. Because $\mapsto_A$ is decidable, we have a characteristic map $\kw{out}  :  \impl{A :  \mathsf{tp} ^+ }.~\tmv{A}  \to  1 + (\mathbb{C}  \times \tmv{A})$. Consider the following functional: 
\[
  \begin{aligned}    
    \Phi _{ \mathsf{eval} }~f~(e, v) =  \begin{cases}   c  \boxplus  f(e', v)
       & \mathsf{out} (e) =  \mathsf{inr}   \cdot  (c, e')  \\   (e = v,  \lambda  -.~0) 
       &  \mathsf{out} (e) =  \mathsf{inl}   \cdot   \star  
      \end{cases} 
  \end{aligned}\]
Define \(\mathsf{eval} : \impl{A :  \mathsf{tp} ^+}.~\tmv{A} \times \tmv{A} \to \kw{T}(1)\) to be the fixed-point of \(\Phi _{ \mathsf{eval} }\) and \(\mathsf{profile}  :  \mathsf{tm} ^+( \mathsf{U} \mathsf{F} 1 )   \to   \mathsf{T} (1)\) as \(\mathsf{profile} (e) =  \mathsf{eval} (e,  \mathsf{ret} ( \star ))\). The meaning of $\kw{eval}_A(e, v) : \kw{T}(1)$ is that when it is defined, $e$ computes to a value $v$ incurring the defined cost. Similarly, $\kw{profile}(e)$ is the cost of computing $e$ when the former is defined. In the following, we will establish some expected properties of the computational semantics such as the uniqueness of evaluation and a ``big-step'' law for sequencing evaluations. 

\begin{restatable}{proposition}{PropEvalFunc}\label[prop]{prop:eval-func}
  The relation $\pi_1 \circ \kw{eval}$ is functional, \ie{} $\supp{\kw{eval}(e, v)}$ and $\supp{\kw{eval}(e, v')}$ implies $v = v'$.
\end{restatable}

\begin{restatable}{proposition}{PropEvalSeq}\label[prop]{prop:eval-seq}
  If \(\kw{eval} (e,  \kw{ret} (v)) = c_1\) and \(\kw{eval} (g~v,  \kw{ret} (w)) = c_2\), then \(\mathsf{eval} (e; g,  \kw{ret} (w)) = c_1 + c_2\). 
\end{restatable}

\begin{corollary}\label[cor]{prop:prof-seq}
If \(\mathsf{eval} (e,  \mathsf{ret} (v)) = c_1\) and \(\mathsf{profile} (g~v) = c_2\), then \(\mathsf{profile} (e; g) = c_1 + c_2\). 
\end{corollary}

In ordinary operational semantics, one has the property that $e; g \Downarrow v$ implies there exists a value $w$ such that $e \Downarrow w$ and $g(w) \Downarrow v$. An analogous rule also holds for the computational semantics, except the value $w$ is not existentially quantified: 

\begin{restatable}{proposition}{PropCompIndSeq}\label[prop]{prop:comp-ind-seq}
  The following inference rule is valid for any $\Sigma$-predicate $\varphi$: 
  \begin{mathpar}
    \inferrule{
      \all{v : A} \kw{eval}(e, v)\supp{} \land \kw{eval}(g~v, \kw{ret}(w))\supp{} \to \varphi(\kw{eval}(e, v) + \kw{eval}(g~v, \kw{ret}(w)))
    }{
      \kw{eval}(e; g, \kw{ret}(w))\supp{} \to \varphi(\kw{eval}(e; g, \kw{ret}(w)))
    }
   \end{mathpar}
\end{restatable}

Using this rule we may derive the following law for profiling compound sequences: 

\begin{restatable}{proposition}{PropProfAssoc}\label[prop]{prop:prof-assoc}
  We have \(\mathsf{profile} ((e; g); i) =  \mathsf{profile} (e; ( \lambda  v.~g~v; i))\).
\end{restatable}

\section{Computational adequacy and noninterference}\label{sec:adequacy}

In this section we show that the computational and denotational semantics satisfy a tight correspondence at the type of \emph{observations}: for every $e : \U{\F{1}}$, we have that $\sem{e}$ is Kleene equivalent to $\kw{profile}(e)$ in the sense that the cost specified computationally and denotationally are equal whenever one of them is defined.%

\subsection{Soundness}

In one direction, it is not too difficult to show \emph{soundness}, which means that the computational steps are respected by the denotational semantics: 

\begin{restatable}{proposition}{PropSoundStep}\label[prop]{prop:sound-step}
  If \(e  \mapsto  c,e'\), then \(\sem{e} = c  \boxplus   \sem{e'}\). 
\end{restatable}

\begin{restatable}{theorem}{ThmSound}\label[thm]{thm:sound}
  If \(\mathsf{eval} (e, v) { \downarrow }\), then $\sem{e} = \kw{eval}(e, v) \costmap \sem{v}$. 
\end{restatable}

\begin{corollary}\label[cor]{coro:sound}
  Given $e : \U{\F{1}}$, we have $\kw{profile}(e) \le \sem{e}$. 
\end{corollary}

\subsection{Adequacy}\label{subsec:adequacy}

Adequacy proper usually refers to the converse direction of the property stated in \cref{coro:sound}: definedness of the denotational semantics implies termination under the computational semantics. Our proof is based on a standard binary logical relations construction between the syntax and semantics of \pcfc{} (\cf{} Plotkin~\cite{plotkin:1977}). The logical relation consists of a family of relations ${\lhd_A} \subseteq \sem{A} \times \tmv{A}$ indexed in the syntactic types $A$ of \pcfc{} such that $\sem{e} \lhd_{\U{\F{1}}} e$ implies the computational adequacy property. The purpose of considering a family of relations is to provide a sufficient strengthening of the desired property to all types so that one may proceed by an inductive proof on the derivation of terms to show that $\sem{e} \lhd_A e$ holds for every term $e : A$. Due to the presence of fixed-point computations, we must show that $- \lhd_{\U{X}} e$ is always an admissible subset of the domain $\sem{X}$ in the sense of \cref{def:admissible}. We define a family of relations called the \emph{formal approximation relations} in \cref{fig:formal-approx-rel} by induction on the structure of syntactic types and show that they satisfy the properties in the preceding discussion.

\begin{figure}
  
  \iblock{
    \begin{multicols}{2}
    \mrow{e  \lhd _{ \kw{1} } e' =  \top}
    \mrow{e  \lhd _{ \kw{ans} } e' = (\overline{e} = e')}
    \mrow{e  \lhd _{ \mathsf{nat} } e' = (e =  \llbracket   e'   \rrbracket )}
    \mrow{e  \lhd _{ \mathsf{U} \mathsf{F} A } e' =  \forall   { \left [ f  \mathrel{( \lhd _A  \Rightarrow   \mathsf{adq} )}  f' \right ]} ~  \mathsf{adq} (f^ \sharp (e), e';f')}
    \mrow{e  \lhd _{ \mathsf{U} (A  \to  X) } e' = e  \mathrel{( \lhd _{A}  \Rightarrow   \lhd _{ \mathsf{U} X })}  e'}
    \columnbreak
    \mrow{\mathsf{adq} (e, e') = (e \le e')}
    \mrow{e  \mathrel{(R  \Rightarrow  S)}  e' =  \forall   { \left [ a  \mathrel{R}  a' \right ]} ~  (e~a)  \mathrel{S}  (e'~a') }
    \end{multicols}
  }
  
   \medskip

   \caption{\emph{Formal approximation relations.} We write $\overline{-} : 2 \to \kw{ans}$ for the function sending $0$ to \kw{no} and $1$ to \kw{yes}. The relation \(\mathsf{adq} \subseteq \kw{T}(1)  \times \tmv{\U{\F{1}}}\) is the ``ground relation'' that generates the formal approximation relations at higher types.}
   \label{fig:formal-approx-rel}
\end{figure}

Formal approximation relations may be extended to open terms as usual. We write $\Gamma \vdash e \lhd_A e'$ when for all closing substitutions $\sigma : \Gamma$, we have that $e(\sem{\sigma}) \lhd_A e'[\sigma]$ holds. The computational adequacy result may be deduced from the \emph{fundamental lemma}: 

\begin{restatable}{theorem}{ThmFLLR} 
  For every closed term $e : \Gamma \vdash A$, the approximation $\Gamma \vdash \sem{e} \lhd_A e$ holds.
\end{restatable}

The proof of the fundamental lemma proceeds by induction on the derivation of terms. Details and the proof that formal approximation predicates $- \lhd_{\U{X}} e$ are admissible can be found in the appendix~\cref{app:adequacy}; crucially we rely on the fact that $\kw{eval}(e, v)$ is a $\Sigma$-proposition. 

\begin{corollary}\label[cor]{coro:adequacy}
  Given $e : \U{\F{1}}$, we have that $\sem{e} = \kw{profile}(e)$. 
\end{corollary}

Extensionally, both the denotational and computational semantics of $e$ are simply partial computations of type $\lift{1}$, so one may view \cref{coro:adequacy} as a cost-sensitive (and internal) version of Ploktin's original adequacy theorem for \pcf{}. Lastly, we see that our semantics of \pcfc{} provides a rigorous proof of the intuitive fact that computations may not observe the cost effect: 

\begin{restatable}{theorem}{ThmNoninterference}
  Any $e : \U{\F{1}} \rightharpoonup \F{2}$ is weakly, extensionally constant in the sense that for all $x, y : \U{\F{1}}$, if $\kw{profile}(x)\supp$ and $\kw{profile}(y)\supp$, then $\kw{eval}(e~x, \kw{ret}(v))\supp$ and $\kw{eval}(e~y, \kw{ret}(u))\supp$ imply $v = u$. 
\end{restatable}

\section{An SDT model of the intension-extension phase distinction}\label{sec:model}

To obtain a model for the constructions of the preceding sections, we instantiate the sheaf model of Sterling and Harper~\cite{sterling-harper:2022} at the poset \(\I  =  { \left \{ \kw{ext}  \le  \kw{int} \right \} }\) representing the intension-extension security order. The basic idea is to first develop domain theory \emph{internal} to the presheaf topos $\PSH{\I}$, from which we may obtain an appropriate \emph{internal} domain-theoretic site that embeds into a sheaf topos model of SDT in the sense of \cref{def:axioms}. The reason to consider internal sites is that we may build into the base category the intension-extension phase distinction that is preserved through the embedding. 

In \opcit{} the internal domain theory of $\PSH{\I}$ is developed in terms of constructive dcpos following the work de Jong~\cite{dejong:2023:thesis}. These internal dcpos are similar to ordinary dcpos; for example, one may use them to give the denotational semantics of \textbf{PCF}~\cite{dejong:2023:thesis}. The difference lies in the dominance $\Sigma$ of the category of internal dcpos given by the subobject classifier $\Omega_{\PSH{\I}}$: while a partial element of ordinary domains in $\textbf{Set}$ is either defined or not, a partial element of a domain internal to $\PSH{\I}$ may have the phase proposition $\P$ as its support, where $\P$ is the intermediate proposition in $\Omega_{\PSH{\I}}$. Recalling the interpretation of a map $A \to \Sigma$ as a computational predicate, this also means that predicates can be phase-dependent in the sense of holding only at the extensional phase. 

We recall some definitions from Sterling and Harper~\cite{sterling-harper:2022} (SH22). 

\begin{definition}
  A \emph{Scott-open immersion} of a dcpo is any mono $U \mono A$ arising from a predicate $A \to \Sigma$. 
\end{definition}

\begin{definition}
  In a category, a \emph{sink} on an object $A$ is a set of morphisms into $A$.  
\end{definition}

\begin{definition}
  In a category with pullbacks, a \emph{Cartesian coverage} is an assignment of objects $A$ to set of sinks on $A$ that is stable under pullback. 
\end{definition}

\begin{definition}
  The \emph{finite open cover topology} is generated by the Cartesian coverage assigning to each dcpo $A$ the set of sinks $\{U_i \mono A\}_i$ on $A$ with every $U_i \mono A$ a Scott-open immersion and $\sup_i U_i \cong A$. 
\end{definition}

Our domain-theoretic site is given by an internal category \(\CatIdent{C}\) of small dcpos in $\PSH{\I}$. We embed \(\CatIdent{C}\) into a Grothendieck topos \Sh{\CatIdent{C}}, obtained by localizing $\PSH{\CatIdent{C}}$ at the finite open cover topology. The purpose of this localization is to ensure that the finite joins of the dominance in \(\CatIdent{C}\) are preserved by the embedding into \Sh{\CatIdent{C}}. This property was notably used by \opcit to implement the semantics of termination declassification; here we use the finite join structure of \(\Sigma\) to show that \(2\) is an extensional predomain. The phase distinction in \Sh{\CatIdent{C}} is inherited from the ambient presheaf topos, where it is represented by $\kw{ext}$. 

\begin{theorem}[SH22, Corollary 90]\label{thm:representable-predomains}
  Every representable presheaf is well-complete.  
\end{theorem}

Thus we obtain a functor $y :  \CatIdent{C} \hookrightarrow \Sh{\CatIdent{C}}$ restricting the Yoneda embedding $\y{\CatIdent{C}} : \CatIdent{C} \hookrightarrow \PSH{\CatIdent{C}}$ onto sheaves. 

\begin{theorem}[SH22, B.1.4.3]
  The representable $\Sigma = y(\Sigma)$ is a dominance in \Sh{\CatIdent{C}}.  
\end{theorem}

\begin{theorem}[SH22, Corollary 79]
  Coproducts in \CatIdent{C} are given by unions of families of Scott-opens.  
\end{theorem}

\begin{corollary}\label[cor]{cor:coproducts-sheaves}
  Finite coproducts of \CatIdent{C} are preserved by the embedding into sheaves. 
\end{corollary}

\begin{theorem}[SH22, Axiom-SDT-1]\label{thm:dominance-finite-join}
  The dominance $\Sigma$ has finite joins that are preserved by the inclusion $\Sigma \hookrightarrow \Omega$.
\end{theorem}

\begin{restatable}{theorem}{ThmModel}\label[thm]{thm:model}
  Setting $\Sigma = y(\Omega_{\PSH{\I}})$ and $\P = y(\y{\I}(\kw{ext}))$, we have that $(\Sh{\CatIdent{C}}, \Sigma, \P)$ is an SDT model of the intension-extension phase distinction in the sense of \cref{def:axioms}. 
\end{restatable}

\begin{proof}
  By \cref{thm:representable-predomains} we know that $\Sigma$ is a well-complete dominance in $\Sh{\CatIdent{C}}$. To show that \Sh{\CatIdent{C}} also models the intension-extension phase distinction, we observe that the presheaf model of the intension-extension phase distinction of Niu~\etal~\cite{niu-sterling-grodin-harper:2022} restricts to a smaller model in \(\CatIdent{C}\): every subterminal object in \PSH{\I} is an internal dcpo, and so we may take the subterminal \(\P  =  \y{\I}(\kw{ext})\) to be the phase proposition in \(\CatIdent{C}\). The phase proposition \(\P =  y(\P) \) in \(\Sh{\CatIdent{C}}\) is classified by $\Sigma$: since $y$ is fully faithful, every $\phi : \Sigma$ arise from a unique map $1_{\PSH{\I}} \to \Omega_{\PSH{\I}}$. 
  
  Moreover, we can directly verify that \(2_{\CatIdent{C}} = 2_{\PSH{\I}} = 1_{\PSH{\I}} + 1_{\PSH{1}}\) is internally orthogonal to \(\P\) in $\CatIdent{C}$. By \cref{cor:coproducts-sheaves}, $y$ preserves finite coproducts, so we have that \(2_{ \Sh{\CatIdent{C}} } =  y(2_\CatIdent{C})\). We observe that 2 is extensional because the restricted embedding \(y :  \CatIdent{C}   \to   \Sh{\CatIdent{C}}\) is both full and faithful and preserves products.
  To see that \(2\) is replete, we observe that 2 is isomorphic to its type of singletons:
  \[\begin{aligned}     
    2  \cong   { \left \{ \phi  :  \Sigma ^2  \mid   { \left ( \all{a, b : 2}   \phi ~a  \land   \phi ~b  \to  a = b \right )}   \land   { \left ( \phi ( \kw{inl} \cdot   \star )  \lor   \phi ( \kw{inr}   \cdot   \star ) \right )} \right \} }    
  \end{aligned}\]
  Because \(\Sigma\) is closed under finite joins (by \cref{thm:dominance-finite-join}), the type of singletons of 2 can be defined as the limit of a diagram of replete types, and so it is replete as well. 
  
  Lastly, to see that Phoa's principle is satisfied, we note that it holds in \(\CatIdent{C}\). Since the subobject \({ \left \{ ( \phi ,  \psi ) :  \Sigma   \times   \Sigma   \mid   \phi   \to   \psi \right \} }\) can be defined as the equalizer of \(\pi _1 :  \Sigma   \times   \Sigma   \to   \Sigma\) and \(\land  :  \Sigma   \times   \Sigma   \to   \Sigma\), every object in the diagram is defined using the cartesian closed structure of \(\CatIdent{C}\), it is preserved by any cartesian closed embedding, and so Phoa's principle also holds in \Sh{\CatIdent{C}}.
\end{proof}

\section{Discussion of related work}\label{sec:discussion}

\emph{Cost analysis in type theory.}
The original motivation for proving internal, cost-sensitive computational adequacy results grew out of the work of Niu \etal{}~\cite{niu-sterling-grodin-harper:2022} on formalizing cost analysis of functional programs in dependent type theory. Niu and Harper~\cite{niu-harper:2023} prove such an adequacy theorem for a variant of the Algol language featuring a notion of first-order recursion in the form of while loops. The purpose of the present paper is to generalize that result to account for higher-order recursion. In contrast to both prior works, we have de-emphasized the role of the call-by-push-value language and instead work directly with the internal language of the SDT topos. Because the model we construct in \cref{sec:model} can also be seen as a model for the type theories in both prior works, we expect that it would be routine to formalize our results in a call-by-push-value version of type theory as well. 

\emph{Cost-sensitive computational adequacy.}
The kind of cost-sensitive adequacy theorem we prove in this paper has been proved in a classic domain-theoretic setting by Kavvos \etal{}~\cite{kavvos-morehouse-licata-danner:2019}, and the general theory of computational adequacy for languages with algebraic effects has been developed by Plotkin and Power~\cite{plotkin-power:2002}. The main difference between our work and those mentioned is that we aim to prove adequacy results internally to a type theory equipped with a phase distinction (such as the one in Niu \etal{}~\cite{niu-sterling-grodin-harper:2022}). As argued in Niu and Harper~\cite{niu-harper:2023}, such internal adequacy theorems can be used as a basis for the validity of axiomatic cost analysis in these type theories. 

\emph{Relative sheaf models of SDT.}
As we have explained in \cref{subsubsec:SDT}, the main sheaf models of synthetic domain theory take the form of Grothendieck topoi over the category of sets. Meanwhile, the logic of phase distinctions finds its home in presheaf topoi in which one finds many distinct subterminal objects that are neither globally true nor globally false; because the category of sets is boolean and two-valued, it can have no non-trivial phase distinctions.
For this reason, Sterling and Harper~\cite{sterling-harper:2022} have proposed to combine synthetic domain theory with phase distinctions by developing models in \emph{relative} Grothendieck topoi~\cite{johnstone:2002} over a presheaf topos that exhibits a phase distinction. In other words, rather than building a site out of predomains in the category of sets, \opcit built an \emph{internal} site based on \emph{internal} predomains in a category of presheaves.
Our model of cost-sensitive synthetic domain theory is similar to that of Sterling and Harper~\cite{sterling-harper:2022}. On the other hand, our proof of computational adequacy is different from that of \opcit, as the latter contains a subtle error~\cite{sterling:2023:noninterference-erratum} involving a mismatch between the existential quantifier and the join of a family of $\Sigma$-propositions in the lifting of free algebras to formal approximation relations.

\emph{Computational adequacy in SDT.} 
Our approach to internal denotational semantics and computational adequacy of \pcfc{} builds on the pioneering work of Simpson~\cite{simpson:1999} on proving the computational adequacy property of \pcf{} in elementary topoi models of SDT. One difference between our work and that of \opcit{} is the addition of the phase distinction, which we use to give an intrinsic denotational account of the interaction between cost and behavior in \pcfc{}. Another difference is in the SDT axioms used and the ensuing definition of the dynamic semantics of the object programming language. Simpson~\cite{simpson:1999} assumes a property called \emph{Axiom N} that closes the dominance $\Sigma$ under countable joins of decidable families in the ambient logic, a fact that we do not rely on in our constructions. The benefit of this axiom is that it enables \opcit{} to give an internal definition of \pcf{} whose dynamic semantics can be characterized by means of existentially quantified statements of the form $\some{n : \Nat} \phi(n)$ where $\phi$ is a primitive recursive predicate. This is used to show that a general property of the internal logic of topoi called \emph{1-consistency}\footnote{A topos $\mathcal{E}$ is \emph{1-consistent} when a closed formula $\some{n : \Nat} \phi(n)$ of the form described above holding in the internal logic of $\CatIdent{E}$ implies that it holds externally.} is both necessary and sufficient to \emph{externalize} the internal adequacy proof into a corresponding proof in ordinary mathematics. In a follow-up paper, Simpson~\cite{simpson:2004} gave a different logical criterion for the equivalence of internal and external adequacy called \emph{computational 1-consistency} that does not rely on Axiom N. Roughly the idea is to define the programming language and its operational semantics in terms of the \emph{computational natural numbers} (analogous to the predomain $\NatP$ in this paper); computational 1-consistency is just the property needed to ensure that internal computational observations hold externally as well.

By contrast, the dynamic semantics of \pcfc{} in this paper is defined \emph{computationally} and is not known to be equivalent to the operational semantics of Simpson~\cite{simpson:1999} in the absence of Axiom N. However, we expect that a version of our computational semantics using the computational natural numbers will be equivalent to the semantics given in Simpson~\cite{simpson:2004}. On the other hand, we find the computational semantics developed in \cref{sec:comp-sem} both philosophically and mathematically compelling and deserving of further investigation in its own right. Moreover, although the computational semantics of \pcfc{} does not appear to be definable in terms of countable joins, it can be defined using \emph{synthetic} $\omega$-joins of decidable families. Therefore, we conjecture that one may externalize the internal adequacy proof of \pcfc{} in the manner of Simpson~\cite{simpson:1999} by developing a Kripke-Joyal semantics for the sheaf model defined in \cref{sec:model} that unfolds an internal statement involving synthetic $\omega$-joins to an external statement in the metatheory. 

\section{Conclusion \& future work}\label{sec:conclusion}

In this paper we study a language for higher-order recursion \pcfc{} in the setting of synthetic domain theory. Our main contribution is an internal, cost-sensitive version of Plotkin's computational adequacy theorem for \pcfc{}. In particular, we define and relate a denotational model of \pcfc{} to a new dynamic semantics for \pcfc{} defined directly in terms of computation that is both natural and mathematically appealing. Here we suggest some ideas for future investigations. 

\emph{Internal \vs{} external adequacy.} 
In the same vein as the work of Simpson~\cite{simpson:1999,simpson:2004}, we are also interested in giving a logical characterization of when internal computational adequacy (with respect to the computational semantics of \pcfc{}) implies external adequacy. However it is not clear to us what would be an analogous condition to 1-consistency: internal notions such as the initial lift algebra $\omega$ and synthetic $\omega$-chains do not have natural external counterparts. As mentioned in \cref{sec:discussion}, a first step would be to develop a systematic understanding of the logical aspects of the initial lift algebra from an external point of view.

As mentioned in \cref{sec:discussion}, one way to obtain external adequacy would be to follow the approach of~\cite{simpson:2004} and define \pcfc{} and its computational semantics purely in terms of computational natural numbers. Alternatively, we may decide to assume Axiom N (see \cref{sec:discussion}), which would imply that the computational semantics coincides with ordinary operational semantics in the internal logic of the SDT topos. We do not expect Axiom N to hold in the model we construct in this paper (\cf{} van Oosten and Simpson~\cite{van-oosten-simpson:2000}), but it does not appear to be a limitation of the general approach to the model construction; indeed we believe it should be possible to start with a different domain-theoretic site such that the embedding into the resulting sheaf topos preserves countable coproducts, which would be enough to validate Axiom N. 

\emph{Cost and information order.}
As discussed in \cref{subsubsec:intrinsic-preorder}, we would like to combine and develop a practical theory for the interaction of the domain-theoretic information order with a cost preorder in the sense of~Grodin \etal{}, who developed a ``preorder'' version of SDT in which ``predomains'' are types equipped with a preorder. Following the approach of relative sheaf models of SDT, we conjecture that one may build a model of SDT that further incorporates an intrinsic preorder structure by starting with a domain-theoretic site internal to the category of \emph{simplicial sets}.%

\emph{Recursive types.} 
We have emphasized recursion at the term level, but synthetic domain theory is also compatible with having recursive types. Simpson~\cite{simpson:2004} has developed in a very general setting the theory and existence of algebraically compact categories of predomains in SDT, and we hope to instantiate the ideas of \opcit{} at a \emph {relative} sheaf model of SDT similar to the one presented in this paper.

\section*{Acknowledgement}
This work was funded by the United States Air Force Office of Scientific Research under grants FA9550-23-1-0728\footnote{\label{PM}Dr.\ Tristan Nguyen, Program Manager}, MURI FA9550-21-0009\cref{PM}, FA9550-21-1-0385\cref{PM} and the National Science Foundation under grant CCF-1901381. We thank Tristan Nguyen at AFOSR for support. Yue Niu was supported by the Air Force Research Laboratory through the NDSEG fellowship. Views and opinions expressed are however those of the authors only and do not necessarily reflect those of AFOSR, AFRL, or NSF.

\bibliographystyle{entics}
\bibliography{bibtex-references/refs-bibtex,bib}

\clearpage
\appendix

\section{Properties of the dominance}

\PropDomLinked*

\begin{proof}
  We need to show that the intrinsic order on \(\Sigma\) coincides with the path relation. First, we observe that \(\bot   \specle   \top\); indeed, fixing a map \(f :  \Sigma   \to   \Sigma\), by Phoa's principle, evaluation at boundary obtains a pair \((f~ \bot , f~ \top )\) such that \(f~ \bot\) implies \(f~ \top\). In one direction, suppose \(x  \pathle  y\), which means we have a path \(l :  \Sigma   \to   \Sigma\) whose boundary is determined by \(x, y\). Fixing \(f :  \Sigma   \to   \Sigma\), we want to show that \(f~x  \to  f~y\). But this follows by evaluating \(\bot   \specle   \top\) at the $\Sigma$-predicate \(f  \circ  l :  \Sigma   \to   \Sigma\). Conversely, if \(x \specle  y\), then we have in particular \(x  \to  y\), which by Phoa's principle uniquely determines a path \(l :  \Sigma   \to   \Sigma\).
\end{proof}

\section{Properties of predomains}

\PropPredomainProperties*
\begin{proof}
  Replete types are both complete and boundary separated because the latter can be defined by orthogonality conditions and is satisfied by $\Sigma$: it is complete by the axioms of SDT (\cref{def:axioms}) and it is boundary separated by \cref{prop:sigma-boundary-sep}. Assuming Phoa's principle, the proof that replete types are linked can be found in Taylor~\cite[Corollary 2.10]{taylor:1991} and Reus~\cite[Corollary 6.1.16]{reus:1995}. Lastly, one can show that the objects whose link relation is antisymmetric can be defined as an orthogonality condition, thus we know that the link relation on every replete type is antisymmetric, hence it follows the intrinsic order on replete types are also partial orders since they are linked.
\end{proof}

\PropSealingSigDecidable*

\begin{proof}
  We observe that $\P \vee A$ is defined as the pushout of the projections of $A \times \P$ as indicated below: 
  \[
  \DiagramSquare{
      ne = A, 
      se = \P \vee A,
      nw = A \times \P,
      se/style = pushout,
      sw = \P, 
      east = \eta_{\P \vee -}, 
      south = \star,
      }
  \]
  Using the fact that $A$ has $\Sigma$-equality, we obtain a map $f : (\P \vee A) \times (\P \vee A) \to \P \vee \Sigma$ such that $f(\eta_{\P \vee -}(x), \eta_{\P \vee -}(y)) = \eta_{\P \vee -}(x = y)$ and $f(\star, -) = f(-, \star) = \star$. The desired characteristic map can then be defined as $\sigma \circ f$, where $\sigma : \P \vee \Sigma \to \Sigma$ is defined as follows:
  \begin{align*}
    &\sigma : \P \vee \Sigma \to \Sigma\\ 
    &\sigma(\eta_{\P \vee -}(\phi)) = \P \vee \phi\\ 
    &\sigma(\star) = \top
  \end{align*}  
  For any $u, v : \P \vee A$, if $u = v = \eta_{\P \vee -}(x)$ for some $x : A$, then we have $\sigma(f(u, v)) = \sigma(\eta_{\P \vee -}(x = x)) = \P \vee (x = x) = \top$. Otherwise, we have that $u$ or $v$ is $\star$, which means $\P$ holds and so $\sigma(f(u, v)) = \sigma(\star) = \top$ as well. Conversely, suppose $\sigma(f(u, v)) = \top$, and that $u = \eta_{\P \vee -}(x)$ and $v = \eta_{\P \vee -}(y)$, which means that $\sigma(\eta_{\P \vee -}(x = y)) = \P \vee (x = y)$ holds. If $\P$ holds, we are done as $(\P \vee A) \cong 1$ in this case. Otherwise, we have $x = y$, and so $u = \eta_{\P \vee -}(x) = \eta_{\P \vee -}(y) = v$. Lastly, if either $u$ or $v$ is the unique element $\star$ then we may discharge the case as above. 
\end{proof}

\PropFuncPointwise*

\begin{proof}
  This is Proposition 5.4.4 of Phoa~\cite{phoa:1991}. We just show the case for the function types. Given a path $f \le_{X \to B} g$, it is clear that we may construct a path $f~x \le B g~x$ for all $x : X$. Conversely, suppose we are given a path $f~x \le B g~x$ for all $x : X$. By \cref{prop:predomain-properties}, $B$ is boundary separated, and so such paths are necessarily unique, and so we have a function $\alpha : X \to \Sigma \to B$ such that $\alpha(x)$ is a path $f~x \le B g~x$. We then obtain a path $f \le_{X \to B} g$ by taking the exponential transpose of $\alpha$.  
\end{proof}

\PropLiftPointwise*

\begin{proof}
  In the forward direction, we have a path \(x{ \downarrow }\) \(\sqsubseteq\) \(y{ \downarrow }\), which means \(x{ \downarrow }\) implies \(y{ \downarrow }\) as the \(\Sigma\) is linked by \cref{prop:dom-linked}. Suppose \(x { \downarrow }\), and let \(f : A  \to   \Sigma\) be arbitrary. By assumption, we have that \(f'(x)  \to  f'(y)\), where \(f'(( \phi , u)) =  \phi   \mathbin{\angle}  f  \circ  u\). In other words, we have \(x { \downarrow }   \mathbin{\angle}  f(x)\) implies \(y { \downarrow }   \mathbin{\angle}  f(y)\). Since \(x { \downarrow }\) holds, we have that \(f(x)  \to  f(y)\), which by definition means \(x  \sqsubseteq _A y\). 
    
  In the backward direction, let \(\alpha  :  x { \downarrow }   \to  x  \sqsubseteq _A y\) be the given partial path. We may define a total path \(\beta  :  \Sigma   \to   \mathsf{L} A\) between \(x\) and \(y\) by setting $\beta ( \phi ) = ( x { \downarrow } ,  \lambda  p.~ \alpha ~p~ \phi )$. Thus we have $x \le_{\lift{A}} y$ as required.
\end{proof}

\PropLub* 

\begin{proof}
  The final \(\mathsf{L}\)-coalgebra is equipped with a global element \(\infty  :  \overline{\omega}\) that can be thought of as the ``point at infinity''. Define \(f_ \infty\) be the element determined by the unique extension \(\overline{f}  :  \overline{\omega}   \to  A\) evaluated at the invariant point \(\infty\). 
  \begin{enumerate}
    \item First we show that \(f_ \infty\) is an upper bound for \(f\). Fixing \(i:  \omega\), we need to show that \(f~i  \sqsubseteq ^ \circ _A f_ \infty\). Because \(\overline{f}\) extends \(f\), it suffices to show \(\overline{f} ~i  \sqsubseteq ^ \circ _A f_ \infty\). Using the fact that every map is monotone with respect to the specialization order, the result holds because \(i  \sqsubseteq ^ \circ _{ \overline{\omega} }  \infty\).
    
    \item Let \(\alpha\) be an upper bound for \(f\). We need to show that \(f_ \infty   \sqsubseteq ^ \circ   \alpha\). If the principal lower set \({ \downarrow }( \alpha )\) is complete, we have the following lifting situation:
    \[\begin{tikzpicture}[diagram]
      \path
      (-1.5,0) node (I) {$\omega$}
      (1.5,0) node (F) {$\overline{\omega}$}
      (0,-2) node (A) [align=center] {${\downarrow}(\alpha) = \{ a \mid a \specle \alpha \}$};
      \draw[embedding] (I) to (F); 
      \draw[->] (I) to node[left] {$f$} (A); 
      \draw[->,exists] (F) to node[right, shift={(0.1,0)}] {$\tilde{f}$} (A);
      \end{tikzpicture}
      \]
      
      In the above \(\tilde{f}\) is the unique extension of \(f\) considered as a map \(\omega   \to  { \downarrow }( \alpha )\). By uniqueness of \(\overline{f}\) as the extension of \(f :  \omega   \to  A\), \(\tilde{f}\) is equal to \(\overline{f}\) considered as maps \(\overline{\omega}   \to  A\). Consequently we have that \(f_ \infty  =  \overline{f} ( \infty ) =  \tilde{f} ( \infty )\), so the result follows by observing that \(\tilde{f} ( \infty )  \in  { \downarrow }( \alpha )\).
    
      It remains to show that \({ \downarrow }( \alpha )\) is complete. We can express the principal lower set as follows:
      \[\begin{aligned}           { \downarrow }( \alpha ) &=  { \left \{ a  \mid  a  \sqsubseteq ^ \circ   \alpha \right \} } \\  
          &=  { \left \{ a  \mid   \forall  f:A \to \Sigma ~ f(a)  \to  f( \alpha ) \right \} } \\  
          &=  \bigcap _{f:A \to \Sigma }  { \left \{ a  \mid  f(a)  \to  f( \alpha ) \right \} }        \end{aligned}\] 
          
      Because complete types are internally complete, the result would follow if we can show that \(S =  { \left \{ a  \mid  f(a)  \to  f( \alpha ) \right \} }\) is complete. We may show that \(S\) can be computed as follows: 
      \[
      \DiagramSquare{
      ne = A, 
      se = \Sigma \times \Sigma,
      nw = S,
      nw/style = pullback,
      sw = \Sigma^\Sigma,
      east = {\langle f, f(\alpha)\rangle}, 
      south = \partial,
      }
      \]
      Since $S$ can be defined as the limit of a diagram of complete types, it is complete as well. 
\end{enumerate}
\end{proof}

\PropMonotone*

\begin{proof}
  Given $a \le b$, we derive a path $\Sigma \to B$ whose boundary is $(f~a, f~b)$ by postcomposing with $f$, and so $f~a \le f~b$ as well. 
\end{proof}

\PropCont* 

\begin{proof}
  Predomains are complete, so we have the following extensions of $d$ and $f \circ d$: 
  \[
  \begin{tikzpicture}[diagram]
    \SpliceDiagramSquare<sq/>{
    width = 3cm,
    nw = \omega, 
    sw = \overline{\omega},
    ne = A, 
    se = B,
    east = f, 
    south = \overline{f \circ d}, 
    west/style = embedding, 
    south/style = {->,exists},
   } 
   \draw[->,exists] (sq/sw) to node[desc] {$\overline{d}$} (sq/ne);
  \end{tikzpicture}
  \]
  Because extensions along $\omega \hookrightarrow \overline{\omega}$ are unique for complete types, we have $f \circ \overline{d} = \overline{f \circ d}$. But by definition of the synthetic $\omega$-join, this means that $f(\sup d) = f(\overline{d}(\infty)) = \overline{f \circ d}(\infty) = \sup(f \circ d)$. 
\end{proof}

\section{Domains and admissibility}

\begin{proposition}\label{prop:path-implies-intrinsic}
  The path relation need not coincide with the intrinsic order, but a path $\alpha : x \pathle y$ always implies $x \specle y$. 
\end{proposition}

\begin{proof}
  Suppose we have a path \(l :  \Sigma   \to   \Sigma\) whose boundary is determined by \(x, y\). Fixing \(f :  \Sigma   \to   \Sigma\), we want to show that \(f~x  \to  f~y\). But this follows by evaluating \(\bot   \specle   \top\) at the $\Sigma$-predicate \(f  \circ  l :  \Sigma   \to   \Sigma\). 
\end{proof}

\begin{proposition}\label{prop:invariant-top}
The invariant point $\infty$ is the top element of $\overline{\omega}$ with respect to the intrinsic preorder. 
\end{proposition}

\begin{proof}
  Observe that for any $i : \overline{\omega}$ we have a path $\alpha : i \pathle \infty$ defined by $\alpha(\phi, n) = \phi \vee i(n)$. Thus we have $i \specle \infty$ by \cref{prop:path-implies-intrinsic}. 
\end{proof}

\BugAdmissible*

\begin{proof}
  Suppose $S \subseteq A$ is downward closed and closed under synthetic $\omega$-joins. We must complete the following lifting problem: 
  \[
  \Local{\omega}{\overline{\omega}}{S}[f][?]
  \]
  Because $S$ is complete, we have the following extension of $\omega \xrightarrow{f} A \xhookrightarrow{\iota} S$:
  \[
  \begin{tikzpicture}[diagram]
    \SpliceDiagramSquare<sq/>{
      nw = \omega,
      sw = S,
      ne = \overline{\omega},
      south/style = {opacity = 0},
      east/style = {opacity = 0},
      north/style = >->,
      west = f,
    }
    \node[below = of sq/sw] (S) {$A$};
    \draw[->,dotted] (sq/ne) to node[right] {$\overline{\iota f}$} (S);
    \draw[->,dotted] (sq/ne) to node[right] {$?$} (sq/sw);
    \draw[>->] (sq/sw) to node {} (S);
  \end{tikzpicture}
  \]
  Thus it suffices to show that $\overline{\iota f}$ factors through $S$. Fixing $i : \overline{\omega}$, we want to show that $\overline{\iota f}(i) \in S$. By the assumptions on $S$, we have that $\sup (\iota f) = \overline{\iota f}(\infty)$ is in $S$ and that $a \in S$ whenever $a \le \overline{\iota f}(\infty)$. Thus it suffices to show $\overline{\iota f}(i) \le \overline{\iota f}(\infty)$, which follows from \cref{prop:invariant-top}. 
\end{proof}

\begin{proposition}
  The intersection of a family of admissible subsets of a domain is admissible.
\end{proposition}

\begin{proof}
  The least element is contained in the intersection as it is contained in every fiber. Suppose that \(f :  \omega   \to  A\) is a synthetic \(\omega\)-chain such that \(f_i  \in   \bigcap  F\). But since \(f_i  \in  F_j\) for every \(j : J\),  we have \(\bigvee  f  \in  F_j\) as every \(F_j\) is admissible, which means \(\bigvee  f  \in   \bigcap  F\) as well.
\end{proof}

\begin{proposition}
  If \(P, Q\) are \(\Sigma\)-subsets of a domain \(A\), then the exponential subobject \(Q^P\) is an admissible subset of \(A\). 
\end{proposition}

\begin{proof}
  Let \(f_i  \in  Q^P\) for some synthetic \(\omega\)-chain \(f\). Suppose that \(\bigvee  f  \in  P\). We need to show that \(\bigvee  f  \in  Q\). By the universal property of $\bigvee$, we may assume that \(f_i  \in  P\) for some \(i :  \omega\). By assumption, this means that \(f_i  \in  Q\) and thence \(\bigvee  f  \in  Q\) as every \(\Sigma\)-subset is monotone in the synthetic order.
\end{proof}

\section{Properties of the denotational semantics} 

\PropCostAlgLaws*

\begin{proof}
  Routine computation using the distributive law $\tau$ of \(\mathbb{C}   \times  -\) over the lifting monad.
\end{proof}

\section{Properties of the computational semantics}

\PropEvalFunc*

\begin{proof}
  Consider the following subset:  \[\begin{aligned}     P =  { \left \{ \alpha   \mid   \forall   { \left [ e, v, v' \right ]} ~   \alpha (e, v) { \downarrow }   \land   \alpha (e, v') { \downarrow }   \to  v = v' \right \} }    \end{aligned}\]  
  
  We may check that \(P\) is admissible by means of \cref{cor:admissible} and proceed by fixed-point induction. Suppose that \(\alpha   \in  P\) and that \(\Phi _ \mathsf{eval} ( \alpha )(e, v) { \downarrow }\) and \(\Phi _ \mathsf{eval} ( \alpha )(e, v') { \downarrow }\). We need to show that \(v = v'\). We proceed by cases on \(\mathsf{out} (e)\).
  \begin{enumerate}
    \item{If \(e \mapsto  c',e'\), we may deduce that \(\alpha (e',v) { \downarrow }\) and \(\alpha (e',v') { \downarrow }\), so the result follows from the assumption that \(\alpha   \in  P\). }    \item{Otherwise, we have that \(e = v\) and \(e = v'\) by definition of \(\Phi _ \mathsf{eval}\), and so \(v = v'\).} 
  \end{enumerate}
\end{proof}

\PropEvalSeq* 

\begin{proof}
  Consider the following subset of \(\Pi _{A :  \mathsf{tp} }.~ \mathsf{U} \mathsf{F} A   \to   \mathsf{U} \mathsf{F} A   \to   \mathsf{T} (1)\):
  \[\begin{aligned}     P =  { \left \{ \alpha   \mid   \forall   { \left [ e \right ]} ~   \alpha (e, v) { \downarrow }   \land   \mathsf{eval} (g~v,  \mathsf{ret} (w)) { \downarrow }   \to   \mathsf{eval} (e; g,  \mathsf{ret} (w)) =  \alpha (e, v) +  \mathsf{eval} (g~v,  \mathsf{ret} (w)) \right \} }    \end{aligned}\]  \par{It suffices to show that \(\mathsf{eval}   \in  P\). Observing that \(P\) is admissible, we proceed by fixed-point induction. Suppose that \(\alpha   \in  P\), \(\Phi _ \mathsf{eval} ( \alpha )(e, v) { \downarrow }\), and that \(\mathsf{eval} (g~v,  \mathsf{ret} (w)) { \downarrow }\). We need to show that \(\mathsf{eval} (e; g,  \mathsf{ret} (w)) =  \Phi _ \mathsf{eval} ( \alpha )(e, v) +  \mathsf{eval} (g~v,  \mathsf{ret} (w))\). We proceed by cases on \(\mathsf{out} (e)\). }  \begin{enumerate}
    \item{      If \(e  \mapsto  c',e'\), then we compute: 
      \[\begin{aligned}          \mathsf{eval} (e; g,  \mathsf{ret} (w)) &= c'  \boxplus   \mathsf{eval} (e'; g,  \mathsf{ret} (w)) \\         &= c' +  \alpha (e', v) +  \mathsf{eval} (g~v,  \mathsf{ret} (w)) \\         &=  \Phi _ \mathsf{eval} ( \alpha )(e, v) +  \mathsf{eval} (g~v,  \mathsf{ret} (w)) 
       \end{aligned}\]      \par{Where the first equality follows from the assumption that \(\alpha   \in  P\) and the second by the definition of \(\Phi _ \mathsf{eval}\). }    }    \item{Otherwise, we have that \(e~ \mathsf{val}\). Since \(\Phi _ \mathsf{eval} ( \alpha )(e,  \mathsf{ret} (v)) { \downarrow }  =  (e =  \mathsf{ret} (v), 0) { \downarrow }  = (e =  \mathsf{ret} (v))\) holds, we can compute: }    \[\begin{aligned}       \mathsf{eval} (e; g,  \mathsf{ret} (w)) &=  \mathsf{eval} ( \mathsf{ret} (v); g,  \mathsf{ret} (w)) \\      &=  \Phi _ \mathsf{eval} ( \mathsf{eval} )( \mathsf{ret} (v); g,  \mathsf{ret} (w)) \\      &=  \mathsf{eval} (g~v,  \mathsf{ret} (w))     \end{aligned}\]    \par{But this is what we needed to show since \(\Phi _ \mathsf{eval} ( \alpha )(e,  \mathsf{ret} (v)) = 0\). }  
\end{enumerate}
\end{proof}

\PropCompIndSeq* 

\begin{proof}
  Consider the subset \(P  \subseteq   \mathsf{U} \mathsf{F} A   \to   \mathsf{U} \mathsf{F} A   \to   \mathsf{T} (1)\) defined as the intersection of the following subsets:
  \[\begin{aligned}     &Q =  { \left \{ \alpha   \mid   \alpha   \sqsubseteq   \mathsf{eval} \right \} }   \\     &R =  { \left \{ \alpha   \mid   \forall   { \left [ c', e', n \right ]} ~  (e  \mapsto ^n c', e')  \land   { \left ( \alpha (e'; g,  \mathsf{ret} (w)) { \downarrow } \right )}   \to   \varphi (c' +  \alpha (e'; g,  \mathsf{ret} (w))) \right \} }    \end{aligned}\]   It suffices to show that \(\mathsf{eval}   \in  P\). We have that \(P\) is admissible, and we proceed by fixed-point induction. Suppose that \(\alpha   \in  P\). We need to show that \(\Phi _ \mathsf{eval} ( \alpha )  \in  P\), where \(\Phi _ \mathsf{eval}\) is the characteristic functional of \(\mathsf{eval}\) (\cref{subsec:comp-sem}). It's immediate that \(\Phi _ \mathsf{eval} ( \alpha )  \in  Q\). It remains to show that it is also contained in \(R\). So suppose that \(e  \mapsto ^n c', e'\) and \(\Phi _ \mathsf{eval} ( \alpha )(e'; g,  \mathsf{ret} (w)) { \downarrow }\). We want to show that \(\varphi (c' +  \Phi _ \mathsf{eval} ( \alpha )(e'; g,  \mathsf{ret} (w)))\). We proceed by cases on \(\mathsf{out} (e')\).  
   \begin{enumerate}
    \item{If \(\mathsf{out} (e') =  \mathsf{inl}   \cdot   \star\), then we know that \(e' =  \mathsf{ret} (v)\) for some \(v :  \mathsf{U} \mathsf{F} 1\). Stepping the operational semantics, we have that \(\mathsf{ret} (v); g  \mapsto  0, g~v\), and by definition of the computational semantics \(\Phi _ \mathsf{eval} ( \alpha )(e'; g,  \mathsf{ret} (w)) = 0  \boxplus   \alpha (g~v,  \mathsf{ret} (w)) =  \alpha (g~v,  \mathsf{ret} (w))\). Since we assumed \(\alpha   \in  Q  \iff   \alpha   \sqsubseteq   \mathsf{eval}\), we also have \(\mathsf{eval} (g~v,  \mathsf{ret} (w)) { \downarrow }\), and since \(\mathbb{C}\) is discrete we have \(\alpha (g~v,  \mathsf{ret} (w)) =  \mathsf{eval} (g~v,  \mathsf{ret} (w))\). Recalling the premise and the fact that \(\mathsf{eval} (e,  \mathsf{ret} (v)) = c'\), we may conclude that \(\varphi (c' +  \mathsf{eval} (g~v,  \mathsf{ret} (w)))\), which is what we needed to show.}    \item{Otherwise, \(\mathsf{out} (e') =  \mathsf{inr}   \cdot  (c'', e'')\) for some \(e'' :  \mathsf{U} \mathsf{F} 1\), and we have that \(e';g  \mapsto  c'', e'';g\). By definition of the computational semantics, this means that \(\Phi _ \mathsf{eval} ( \alpha )(e'; g,  \mathsf{ret} (w)) = c''  \boxplus   \alpha (e''; g,  \mathsf{ret} (w))\). Since we assumed that \(c''  \boxplus   \alpha (e''; g,  \mathsf{ret} (w)) { \downarrow }\), we can use the laws of the derived algebra (\cref{prop:cost-alg-laws}) to deduce that \(\alpha (e''; g,  \mathsf{ret} (w)) { \downarrow }\) as well, and so by the assumption that \(\alpha   \in  R  \subseteq  P\), we have that \(\varphi (c' + c'' +  \alpha (e'';g,  \mathsf{ret} (w)))\) holds, which is what we needed to show. }  
\end{enumerate} 
\end{proof}

\PropProfAssoc* 

\begin{proof}
  In one direction, we show that \(\mathsf{profile} ((e; g); i) { \downarrow }\) implies \(\mathsf{profile} (e; ( \lambda  v.~g~v; i)) { \downarrow }\) and both denote identical costs. Consider the \(\Sigma\)-predicate \(\varphi\) such that \(\varphi (c)\) if and only if \(\mathsf{profile} (e; ( \lambda  v.~g~v; h)) = c\). Suppose that \(\mathsf{eval} (e; g,  \mathsf{ret} (w)) { \downarrow }\) and \(\mathsf{profile} (i~w) { \downarrow }\). By computational induction on sequencing~\cref{prop:comp-ind-seq}, it suffices to show that \(\mathsf{profile} (e; ( \lambda  v.~g~v; h)) =  \mathsf{eval} (e; g,  \mathsf{ret} (w)) +  \mathsf{profile} (i~w)\). Applying computational induction on \(\mathsf{eval} (e; g,  \mathsf{ret} (w)) { \downarrow }\), we further suppose that \(\mathsf{eval} (e,  \mathsf{ret} (v)) { \downarrow }\) and \(\mathsf{eval} (g~v,  \mathsf{ret} (w)) { \downarrow }\) and aim to show that \(\mathsf{profile} (e; ( \lambda  v.~g~v; h)) =  \mathsf{eval} (e,  \mathsf{ret} (v)) +  \mathsf{eval} (g~v,  \mathsf{ret} (w)) +  \mathsf{profile} (i~w)\). 
  \begin{enumerate}
    \item{We claim that \(\mathsf{profile} (e; ( \lambda  v.~g~v; i)) { \downarrow }\). By the big-step semantics of profiling~\cref{prop:prof-seq}, it suffices to show that \(\mathsf{eval} (e,  \mathsf{ret} (v)) { \downarrow }\) for some \(v\) and \(\mathsf{profile} (g~v; i) { \downarrow }\). The former follows from our assumption; for the latter, it suffices to show that \(\mathsf{eval} (g~v,  \mathsf{ret} (w)) { \downarrow }\) and \(\mathsf{profile} (i~w) { \downarrow }\), both of which follow from assumptions. }    \item{Given that \(\mathsf{profile} (e; ( \lambda  v.~g~v; i)) { \downarrow }\), we may apply computational induction again: supposing that \(\mathsf{eval} (e,  \mathsf{ret} (v')) { \downarrow }\) and \(\mathsf{profile} (g~v; i) { \downarrow }\), we have to show that \(\mathsf{eval} (e,  \mathsf{ret} (v')) +  \mathsf{profile} (g~v; i) =  \mathsf{eval} (e,  \mathsf{ret} (v)) +  \mathsf{eval} (g~v,  \mathsf{ret} (w)) +  \mathsf{profile} (i~w)\), which follows from the uniqueness of evaluation~\cref{prop:eval-func} and big-step semantics of profiling~\cref{prop:prof-seq}. 
    }  
\end{enumerate}  \par{In the other direction, suppose that \(\mathsf{profile} (e; ( \lambda  v.~g~v; i)) { \downarrow }\). It suffices to show that \(\mathsf{profile} ((e; g); i) { \downarrow }\). By computational induction, we may assume that \(\mathsf{eval} (e,  \mathsf{ret} (v)) { \downarrow }\) and \(\mathsf{profile} (g~v; i) { \downarrow }\). Applying computational induction again, we can also assume that \(\mathsf{eval} (g~v,  \mathsf{ret} (w)) { \downarrow }\) and \(\mathsf{profile} (i~w) { \downarrow }\) for some \(w\). By the big-step semantics of profiling~\cref{prop:prof-seq}, it suffices to show that \(\mathsf{eval} (e; g,  \mathsf{ret} (w)) { \downarrow }\) and \(\mathsf{profile} (i~w) { \downarrow }\). The latter is our assumption, and the former follows from the big-step semantics of evaluation~\cref{prop:eval-seq}.}
\end{proof}

\begin{restatable}{proposition}{PropCompIndApp}\label[prop]{prop:comp-ind-app}
  The following is valid: 
  \begin{mathpar}
  \inferrule{
    \all{e} \kw{eval}(f, \lambda e)\supp{} \land \kw{eval}(e[v], \kw{ret}(w))\supp{} \to \varphi(\kw{eval}(f, \lambda e) + \kw{eval}(e[v], \kw{ret}(w)))
  }{
    \kw{eval}(f~v, \kw{ret}(w))\supp{} \to \varphi(\kw{eval}(f~v, \kw{ret}(w)))
  }
  \end{mathpar}
\end{restatable}

\begin{restatable}{proposition}{PropCommAppSeq}\label[prop]{prop:comm-app-seq}
  We have that \(\mathsf{eval} ((e; g)~w, z) =  \mathsf{eval} (e;  \lambda  v.~g~v~w, z)\). 
\end{restatable}

\section{Soundness of the denotational semantics}

\PropSoundStep* 

\begin{proof}
  By induction on the derivation of \(e  \mapsto  c,e'\). 
\end{proof}

\ThmSound* 

\begin{proof}
  Consider the following subset:
  \[\begin{aligned}     P =  { \left \{ \alpha   \mid   \forall   { \left [ e \right ]} ~   \alpha (e, v) { \downarrow }   \to   \llbracket   e   \rrbracket  =  \mathsf{eval} (e, v)  \boxplus   \llbracket   v   \rrbracket \right \} }  
   \end{aligned}\]  \par{Because \(\llbracket   e   \rrbracket  =  \mathsf{eval} (e, v)  \boxplus   \llbracket   v   \rrbracket\) is a \(\Sigma\)-proposition, we see that \(P\) is an admissible subset. Suppose that \(\alpha   \in  P\) and \(\Phi _{ \mathsf{eval} }( \alpha )(e, v) { \downarrow }\). We need to show that \(\llbracket   e   \rrbracket  =  \mathsf{eval} (e, v)  \boxplus   \llbracket   v   \rrbracket\). We proceed by cases on \(\mathsf{out} (e)\).}  \begin{enumerate}
    \item{If \(e  \mapsto  c',e'\), then by the soundness of the one step relation~\cref{prop:sound-step}, it suffices to show that \(c'  \boxplus   \llbracket   e'   \rrbracket  = c'  \boxplus   \mathsf{eval} (e', v)  \boxplus   \llbracket   v   \rrbracket\), which follows from the assumption, noting that \(\Phi _{ \mathsf{eval} }( \alpha )(e, v) { \downarrow }\) implies \(\alpha (e, v) { \downarrow }\).}    \item{Otherwise, we have that \(e = v\), and so the result holds since \(\mathsf{eval} (e, e) = 0\).}  
\end{enumerate}
\end{proof}

\section{Proofs for the computational adequacy property}\label{app:adequacy}

\begin{proposition}\label[prop]{prop:head-exp}
  If \(d  \lhd _X e\) and \(e'  \mapsto  c, e\), then \(c  \boxplus  d  \lhd _X e'\). 
\end{proposition}

\begin{proof}
  By induction on \(X\), using the laws of the cost algebra\cref{prop:cost-alg-laws}. 
\end{proof}

\subsection{Admissibility}

\begin{proposition}\label[prop]{prop:free-adm}
  We have that \(-  \lhd _{ \U{\F{A}} } e\) is an admissible subset of \(\kw{T}(A)\).
\end{proposition}

\begin{proof}
  By \cref{cor:admissible} it suffices to show downward closure and closure under $\bot$ and synthetic $\omega$-joins. 
  \begin{enumerate}
    \item By definition this means to show $f^\sharp(\bot) \le \kw{profile}(e; g)$ for all \(f  \mathrel{(A \Rightarrow \kw{adq})} g\). But this holds since $f^\sharp(\bot) = \bot$ and $\bot$ is the least element of $\kw{T}(A)$.
    \item Let $d$ be a synthetic $\omega$-chain such that $d_i \lhd_{\U{\F{A}}} e$ for all $i : \omega$. We want to show that $\sup d \lhd_{\U{\F{A}}} e$, which is to show that $f^\sharp(\sup d) \le \kw{profile}(e; g)$ for all \(f \mathrel{(A \Rightarrow \kw{adq})}  g\). Since $f^\sharp(\sup d) = \sup(f^\sharp \circ d)$, this means to show $\sup(f^\sharp \circ d) \le \kw{profile}(e; g)$. By the universal property of the synthetic $\omega$-join, it suffices to show $f^\sharp(d_i) \le \kw{profile}(e; g)$ for all $i:\omega$, but this is the assumption.
    \item Fixing $d' \le d \lhd_{\U{\F{A}}} e$, we need to show that $d' \lhd_{\U{\F{A}}} e$. Suppose that \(f \mathrel{(A \Rightarrow \kw{adq})}  g\). We need to show $f^\sharp(d') \le \kw{profile}(e; g)$. By the characterization of the order on lifts \cref{prop:lift-order}, we suppose $f^\sharp(d')\supp$ and show that $\kw{profile}(e; g)\supp$ and that $f^\sharp(d') = \kw{profile}(e; g)$. By assumption we know $d' = \eta_{\mathbb{L}}(a')$ and $f(a') = c$ for some $a : \sem{A}$ and $c : \P \vee \mathbb{C}$. Since $d' \le d$, we know $d' = \eta_{\mathbb{L}}(a)$ for some $a$ such that $a' \le a$. Consequently, we have $c = f(a') \le f(a)$, but since $\P \vee \mathbb{C}$ is discrete (by \cref{sec:den-sem} and the argument in \cref{prop:comp-ind-seq}), we have $f(a) = c = f(a')$. Thus by the assumption that $d \lhd_{\U{\F{A}}} e$, we have $f^{\sharp}(d) = c \le \kw{profile}(e; g)$. Again by the discreteness of $\P \vee \mathbb{C}$ we have that $f^\sharp(d') = f(a') = c = \kw{profile}(e; g)$, as required.
  \end{enumerate}
\end{proof}

\begin{proposition}\label[prop]{prop:sup-func}
  Suprema of synthetic $\omega$-chains in function spaces are computed pointwise.
\end{proposition}

\begin{proposition}\label[prop]{prop:func-adm}
  Given that \(-  \lhd _{ \mathsf{U} X } e\) is admissible for all \(e :  \mathsf{U} X\), we have that \(-  \lhd _{ \mathsf{U} (A  \to  X) } e\) is as well.
\end{proposition}

\begin{proof}
  Again we show downward closure and closure under $\bot$ and $\sup$. 
  \begin{enumerate}
    \item Because \(\bot (a) =  \bot\), we have that \(\bot   \lhd _{ \mathsf{U} (A  \to  X) } e\) by the assumption that $\bot \lhd_{\U{X}} e$ for all $e$. 
    \item Suppose that \(f_i  \lhd _{ \mathsf{U} (A  \to  X) } e\). We need to show that \(\bigvee  f  \lhd _{ \mathsf{U} (A  \to  X) } e\). Suppose that \(a  \lhd _A b\). We need to show that \(( \bigvee  f)~a  \lhd _{ \mathsf{U} X } e~b\). This follows the fact that synthetic \(\omega\)-joins in function spaces are computed pointwise and the assumption that \(-  \lhd _{ \mathsf{U} X } e~b\) is closed under $\sup$. 
    \item Fix $f' \le f \lhd_{\U{(A \to X)}} e$. To show that $f' \lhd_{\U{(A \to X)}} e$, suppose that $a \lhd_A b$. We need to show that $f'~a \lhd_{\U{X}} e~b$. By the premise, we have that \(-  \lhd _{ \mathsf{U} X } e~b\) is a lower set, so it suffices to show $f'~a \le f~a \lhd_{\U{X}} e~b$, which follow from the assumptions $f' \le f$ and $f \lhd_{\U{(A \to X)}} e$. 
  \end{enumerate} 
\end{proof}

\begin{proposition}\label[prop]{prop:formal-adm}
  Given \(e :  \U{X}\), we have that \(-  \lhd _{ \U{X} } e\) is an admissible subset of \(\sem{\U{X}}\). 
\end{proposition}

\begin{proof}
 By \cref{prop:free-adm,prop:func-adm}.  
\end{proof}

\subsection{Fundamental lemma}

We give the representative cases of the proof by induction on the derivation of terms. 

\begin{lemma}\label[lem]{lem:fllr-ret}
  If \(a  \lhd _A v\), then \(\eta _ \mathsf{T} (a)  \lhd _{ \mathsf{U} \mathsf{F} A }  \mathsf{ret} (v)\). 
\end{lemma}

\begin{proof}
  Let \(f  \mathrel{( \lhd _A  \Rightarrow {\kw{adq}} )}  g\). We need to show that \((f^ \sharp ( \eta _ \mathsf{T} (a)))  \mathrel{\kw{adq}}  ( \mathsf{ret} (v); g)\). Computing the denotational semantics and applying \cref{prop:head-exp}, it suffices to show that \((f~a) \mathrel{\kw{adq}} (g~v)\), which follows from our assumption. 
\end{proof}

\begin{lemma}\label[lem]{lem:fllr-seq}
  If \(d  \lhd _{ \mathsf{U} \mathsf{F} A } e\) and \(f  \lhd _{ \mathsf{U} (A  \to  X) } g\), then \(f^ \sharp (d)  \lhd _{ \mathsf{U} X } e; g\).
\end{lemma}

\begin{proof}
  \par{By induction on \(X\).}  \begin{enumerate}
    \item{If \(X =  \mathsf{F} B\), let \(h  \mathrel{( \lhd _B  \Rightarrow   \mathsf{adq} )}  i\). We need to show that \(h^ \sharp (f^ \sharp (d))  \mathrel{\mathsf{adq}}  (e; g); i\). Computing the denotational semantics and using the fact that we may reassociate sequences (\cref{prop:prof-assoc}), it suffices to show \(((h^ \sharp   \circ  f)^ \sharp (d))  \mathrel{\mathsf{adq}}  (e; ( \lambda  v.~g~v; i))\). By the assumption that \(d  \lhd _{ \mathsf{U} \mathsf{F} A } e\), it suffices to show that for all \(a  \lhd _A v\), we have that \((h^ \sharp (f~a))  \mathrel{\mathsf{adq}}  (g~v; i)\), which follows directly from the assumptions that \(f  \lhd _{ \mathsf{U} (A  \to  X) } g\) and \(h  \mathrel{( \lhd _B  \Rightarrow   \mathsf{adq} )}  i\). 
    }    \item{If \(X = B  \to  Y\), suppose that \(b  \lhd _B v\). We need to show that \((f^ \sharp (d))~b  \lhd _{ \mathsf{U} Y } (e; g)~v\). Unraveling the denotational semantics and the computational semantics (using \cref{prop:comm-app-seq}), it suffices to show \(( \lambda  d.~f~d~b)^ \sharp ~d  \lhd _{ \mathsf{U} Y } (e;  \lambda  d.~g~d~v)\), which follows from the inductive hypothesis and the assumption that \(f  \lhd _{ \mathsf{U} (A  \to  (B  \to  Y)) } g\). }
\end{enumerate}
\end{proof}

\begin{lemma}\label[lem]{lem:fllr-step}
  If \(d  \lhd _X e\), then \(c  \boxplus  d  \lhd _X  \mathsf{step} ^c(e)\). 
\end{lemma}

\begin{proof}
  Since \(\mathsf{step} ^c(e)  \mapsto  c, e\), the result holds by \cref{prop:head-exp}.
\end{proof}

\ThmFLLR* 

\begin{proof}
  By \cref{lem:fllr-ret,lem:fllr-seq,lem:fllr-step,prop:formal-adm}. 
\end{proof}

\ThmNoninterference* 

\begin{proof}
  Let $c : \mathbb{C}$ and $d : \mathbb{C}$ be the costs denoted by $\kw{eval}(e~x, \kw{ret}(v))$ and $\kw{eval}(e~y, \kw{ret}(u))$. By soundness \cref{thm:sound} and laws of the derived algebra \cref{prop:cost-alg-laws}, we have that $\sem{e~x} = c \costmap \eta_\mathbb{T}(\sem{v}) = \eta_\mathbb{L}(c, \sem{v})$ and similarly $\sem{e~y} =\costmap \eta_\mathbb{T}(\sem{u}) = \eta_\mathbb{L}(d, \sem{u})$. It suffices to show that $\sem{v} =_2 \sem{u}$. Because 2 is a purely extensional type (as required by \cref{def:axioms}), we may assume that $\P$ holds. By assumption and soundness, we have $\sem{x} = \sem{y}$, and so $\sem{e~x} = \eta_\mathbb{L}(c, \sem{v}) = \eta_\mathbb{L}(d, \sem{u}) = \sem{e~y}$, which means that $\sem{v} = \sem{u}$ since $c = d$ as elements of a purely intensional type $\mathbb{C}$. 
\end{proof}

\end{document}